\title{Nonparametric causal inference from observational time
  series through marginal integration} 
\author{Shu Li\thanks{
                  Seminar for Statistics, ETH Z\"urich, last name @ stat.math.ethz.ch}\hspace{0.15cm}\thanks{Supported by SNF 2-77991-14}
	\and
        Jan Ernest\footnotemark[1]\hspace{0.15cm}\thanks{Supported in part by the Max Planck ETH Center for Learning Systems}
        \and
        Peter B\"uhlmann\footnotemark[1]
        }
\documentclass{article}
\usepackage{graphicx}
\usepackage{amsmath}
\usepackage{amssymb}
\usepackage{amsthm}
\usepackage[longnamesfirst,numbers]{natbib}
\usepackage[top=1in, bottom=1.25in, left=1.25in, right=1.25in]{geometry}

\usepackage{color}
\usepackage{subcaption}
\usepackage{caption}
\usepackage{algorithm}
\usepackage{algorithmic}
\usepackage{url}

\newtheorem{theorem}{Theorem}
\newtheorem{lemma}{Lemma}
\newtheorem{assumption}{Assumption}

\newtheorem{definition}{Definition}
\newtheorem{remark}{Remark}

\newcommand{\E}{\mathbb{E}}

\newcommand{\ARMA}{\mathrm{ARMA}}
\newcommand{\GARCH}{\mathrm{GARCH}}

\newcommand{\AR}{\mathrm{AR}}
\newcommand{\ARCH}{\mathrm{ARCH}}

\newcommand{\R}{\mathbb{R}}
\newcommand{\N}{\mathbb{N}}
\newcommand{\Z}{\mathbb{Z}}
\newcommand{\Superset}{\mathcal{S}}
\newcommand{\Cov}{\mathrm{Cov}}
\newcommand{\Var}{\mathrm{Var}}

\newcommand{\supp}{\mathrm{supp}}

\newcommand{\Do}{\mathrm{do}}

\newcommand{\pa}{\mathrm{pa}}

\newcommand{\argminn}{\operatornamewithlimits{argmin}}

\begin{document}
\newpage
\maketitle
\begin{abstract}
Causal inference from observational data is an ambitious but highly relevant task,
with diverse applications ranging from natural to social sciences.
Within the scope of nonparametric time series, 
causal inference defined through interventions (cf. Pearl~\citep{pearl}) is largely unexplored, although time order simplifies the problem
substantially. 
We consider a marginal integration scheme for
inferring causal effects from observational time series data, \mbox{MINT-T}
(\textbf{m}arginal \textbf{int}egration in \textbf{t}ime series), which is
an adaptation for time series of a method proposed by Ernest and B\"{u}hlmann (Electron. J. Statist, pp. 3155-3194, vol. 9, 2015)
for the case of independent data. Our approach for stationary stochastic
processes is fully nonparametric and, assuming no instantaneous
  effects consistently recovers the total
causal effect 
of a single intervention with optimal one-dimensional nonparametric
convergence rate $n^{-2/5}$ assuming regularity conditions and twice differentiability of a certain corresponding 
regression function. Therefore, \mbox{MINT-T}
remains largely unaffected by the curse of dimensionality as
  long as smoothness conditions hold in higher dimensions and it is
feasible for a large class of stationary time series, including nonlinear
and multivariate processes. For the case with instantaneous effects,
  we provide a procedure which guards against false positive causal
  statements. 
\end{abstract}  


\section{Introduction}

In a time series setting, causal reasoning revolves predominantly around
Granger causality \cite{granger1969}. Roughly speaking,
a time series is Granger causal to 
another one if knowing the past of the former helps predict the
future of the latter, given all other available
  information of the past (e.g., from the second or other time
  series). This is a concept of "predictive
    causality''. Furthermore, Granger causality is 
  measuring a direct effect, e.g., its
  targets are the entries of the coefficient matrix in a vector
  autoregressive model. Here, we are considering total causal effects
  describing the total effect of an intervention at a (single) variable:
  although less ambitious than direct effects, we will argue here that
  total effects are much more feasible to infer in a model-free,
  nonparametric way. For this task, we use the framework of causal  
reasoning through interventions as described by \citet{pearl} or 
Spirtes et al.~\cite{Spirtes2000}. It is largely unexplored for the case of time
series,
although time dependence has plenty to offer in this particular setting. 
From our perception of time, the present is affected by the past but not vice
versa. 
We also commonly note that causes precede their effects. Both
relations are inherently asymmetric. In the simplest setting of two
correlated random variables without hidden confounders, we are not able to tell
apart the cause from its effect without making any additional assumptions
in the i.i.d.\ setting. This is in contrast to time series where in
  general we can distinguish cause from effect by 
looking at the order in time thereby assuming that the time resolution of the
  measurements is higher than the timescale of causal influences, that is,
  there are no instantaneous effects \cite{hyvarinen2008causal}. In case of
  instantaneous effects, one can still derive some interesting statements
  as we discuss in Section \ref{sec:instant}. Time ordering is
a simplification of a more general phenomenon, though. 
In general, the estimation of causal models in the i.i.d.\ case involves
finding a causal ordering of random variables, and this often involves
unverifiable assumptions such as faithfulness and sophisticated structure
search algorithms such as GES~\citep{chick02}, PC~\citep{Spirtes2000} or
CAM~\citep{buhlmann2014cam}. Moreover, without making further assumptions
on the data-generating process, the causal model is typically only
identifiable up to an equivalence class of valid causal orderings. Within
time series, we can often ignore the identifiability aspects and skip the
complicated part of structure learning by simply propagating time as the
causal order.  

We consider the average total causal effect which is defined through Pearl's $\Do$-operator~\citep{pearl}:
\begin{eqnarray}\label{causaleff}
\E[X_t\,\vert\, \Do(X_{t-s}=x)],
\end{eqnarray}
where the $\Do$-operator encodes an external intervention by setting the
random  variable $X_{t-s}$ to the fixed deterministic value $x$; here $X_t$
and $X_{t-s}$ are real-valued. We often
refer to this quantity as the causal effect of $X_{t-s}$ on $X_t$ (and suppress the words "average total"). For a
stationary time series $({X}_t)_{t \in \Z}$, the causal effect is invariant
under time shifts in $t$. In case of a multivariate time series, we would
consider components $X_{c_1,t}$ and $X_{c_2,t-s}$ in
\eqref{causaleff}. The quantity in \eqref{causaleff} is a general function
of $x$, for each value of $s$ (and in the multivariate case it also depends on $c_1$ and $c_2$). Our goal is to estimate this function of $x$ in a
fully nonparametric way without relying on a specific model specification
for the underlying time series. 

A simple example should help to illustrate our goal. Consider a stationary
$\AR(p)$-model: $X_t=\sum_{j=1}^p \phi_j X_{t-j}+\epsilon_t$, where $p$ is
the order of the Markovian process, and we assume the $\epsilon_t$'s
to be i.i.d. with mean zero and $\epsilon_t$ to be independent of
  $\{X_s;\ s< t\}$. Then,
\begin{eqnarray*}
\E[X_t\,\vert\, \Do(X_{t-s}=x)]&=&\sum_{j=1}^p\phi_j
                                      \E[X_{t-j}\,\vert\,
                                      \Do(X_{t-s}=x)]. 
\end{eqnarray*}
\begin{sloppypar}
For a fixed intervention value $x$, the causal effect can be calculated
recursively via $g_x(s):=\E[X_t\,\vert\,
\Do(X_{t-s}=x)]=\sum_{j=1}^p\phi_j g_x(s-j)$ with initial values
$g_x(0)=\E[X_t\,\vert\, \Do(X_{t}=x)]=x$ and $g_x(-s)=\E[X_t\,\vert\,
\Do(X_{t+s}=x)]=\E[X_t]=0$ for $s>0$. Since $g_x(s)$ is a
linear combination of previous $\{g_x(s-i)\,\vert\, i>0\}$, we thus see that
the causal effect $\E[X_t\,\vert\,\Do(X_{t-s}=x)]$ is a linear function of the
intervention value $x$. Furthermore, $\E[X_t\,\vert\, \Do(X_{t-s}=x)]$ 
is geometrically decaying because $g_x(s)$ satisfies the skeleton equation of
the stationary AR(p)-model. However, if the true underlying data generating
stationary process is nonlinear, perhaps with non-additive innovation
terms, the causal effect is typically a nonlinear function of $x$ and such a recursive formulation is difficult and not useful for the
task of estimating the causal effect. \end{sloppypar}

Instead, we will use an approach
based on Pearl's backdoor adjustment formula~\citep{pearl}. The
interventional density can be calculated from the observational density
\begin{equation}
\label{eq:backdoor}
p(x_{t}\,\vert\, \Do(X_{t-s}=x))=\int p(x_{t}\,\vert\, X_{t-s}=x, X_{t-s}^\Superset)dP( X_{t-s}^\Superset),
\end{equation}
where $X_{t-s}^\Superset$ denotes a so-called adjustment set of $X_{t-s}$, and
the expression in \eqref{eq:backdoor} is invariant under shifts in $t$ due
to the stationarity assumption. In a
time series setting with a data-generating univariate Markovian
process, the adjustment set 
can be chosen as $X_{t-s}^\Superset = \{X_{t-s-1},\ldots,X_{t-s-p}\} $
, the previous instances in time 
up to a time lag $p\in \N^0$ where the value $p$ has to be at least as
large as the order of the Markovian process. Thus, for stationary
Markovian processes we obtain: 
\begin{eqnarray}\label{backdoor2}
\E[X_{t}\,\vert\, \Do(X_{t-s}=x)]=\int \E[X_{t}\,\vert\, X_{t-s}=x, X_{t-s}^\Superset] dP(X_{t-s}^\Superset).
\end{eqnarray}
In particular, \eqref{backdoor2} illustrates that we can relate the causal
effect to integrating the regression function $m(x,X_{t-s}^\Superset) =
\E[X_{t}\,\vert\, X_{t-s}=x, X_{t-s}^\Superset]$ over the adjustment
variables $X_{t-s}^\Superset$, and this is the key property to construct an
estimator in a model free way. This line of reasoning can be extended to
multivariate $l$-dimensional time series where causal statements are often more
interesting: we aim to estimate 
\begin{eqnarray*}
\E[\mathbf{X}_{t}\,\vert\, \Do(X_{c,t-s}=x)],
\end{eqnarray*}
where $c \in \{1,\ldots ,l\}$ denotes a single component. The same formula
\eqref{backdoor2} applies when using an appropriate adjustment set $\mathbf{X}_{t-s}^\Superset$. The
latter will depend whether there are time-instantaneous effects across the
different components or not.
We show that when having no instantaneous effects, assuming regularity
conditions and 
twice differentiability of $\E[\mathbf{X}_{t}\,\vert\, X_{c,t-s}=x,
\mathbf{X}_{t-s}^\Superset]$ with respect to $x$, our proposed marginal
integration estimator achieves the $n^{-2/5}$ convergence rate for
estimating the true causal effect $\E[\mathbf{X}_{t}\,\vert\,
\Do(X_{c,t-s}=x)]$. In case of instantaneous effects, we will
derive a conservative procedure which does not require any knowledge about
the true underlying structure of the Markovian process.

\subsection{Related work and our contribution}

Nonparametric estimation of causal effects of continuous
(intervention/treatment) random variables
with marginal integration of the regression function has been proposed by
\citet{ernest2015} in the 
i.i.d.\ setting. The difficulty there is that the adjustment in
\eqref{eq:backdoor} has to be estimated from data, unless it is known. This
task is rather delicate, as it amounts to determining a superset of the 
parents of the intervention variable in a structural equation model; and it
is difficult to imagine that this could be accurately done in a
nonparametric way. 

Here, with Markovian time series and without instantaneous effects
  in multivariate settings, a valid adjustment
set only needs the knowledge or a good estimate of an upper bound of the
Markovian order. This makes our procedure much more reliable, and in fact,
our implementation of marginal integration is somewhat different and more
direct than the proposal in \citet{ernest2015}. Our investigation
  for the case with instantaneous effects in multivariate scenarios is novel
  and specific to time series problems. Regarding the mathematical
analysis, the technical derivation of the convergence rate for estimating causal effects is 
substantially more demanding and we extend the existing theory to strongly
mixing stationary multivariate processes. Thereby, a fully
  nonparametric Markovian process setting with smoothness and additional
  regularity conditions is assumed. Earlier work on the asymptotic theory
  of marginal integration in a time series setting has been considered in
 \citet{tjostheim1994}, \citet{masry1995} and \citet{masry1997}. These contributions consider special forms of a
  nonparametric time series model, i.e., univariate, ARCH- and ARX-models, whereas our focus is on a more model free
  approach. It is fair to say though that our theoretical analysis shares
  some similarities with the previous works mentioned above.

There is a substantial literature on nonparametric estimation of causal
effects of binary treatment or intervention variables
\citep{robins2009,li2011}. For continuous treatments or intervention
variables, recent work also includes Kennedy et al.~\cite{kennedy2015}, besides the marginal
integration approach \citep{ernest2015}. 
Marginal integration itself has been proposed by \citet{linton1995} for smooth
function estimation in structured nonparametric regression, mainly for
additive models. The theoretical analysis in \citet{fan1998} is the basis
for inferring causal effects in the setting of independent data
\citep{ernest2015}. We are not aware of any other work which considers
estimation of causal effects, as defined in \eqref{causaleff}, for the
setting of stationary time series. 

\section{Marginal integration for time series}
\label{section:theory}

The marginal integration regression method was first proposed by
\citet{linton1995} in the context of additive regression modelling. It is
based on the idea that each component of an additive model can be obtained
by weighted marginal integration of the regression function. Adopting the
notation from regression, let $Y$ be a real-valued response variable, $X_1
\in \R$ and $X_2 \in \R^q$ a set of random (co-)variables. We denote the
regression function $\E[Y\,\vert\,X_1=x_1,\,X_2=x_2]$ by $m(x_1,x_2)$.
\citet{fan1998} show that the one-dimensional function of $x_1$ when
marginalised over $X_2$ 
, i.e., $\E[W(X_2) m(x_1,X_2)]$, can be consistently estimated by weighted
marginal empirical integration of 
the estimated regression function with the one-dimensional
nonparametric convergence rate under certain smoothness
conditions. Here, $W(\cdot)$ is a real-valued weight function 
satisfying $\E[W(X_2)]=1$. 
The presence of the additional variables $X_2$ does not add complexity to
the estimation of  $\E[W(X_2) m(x_1, X_2)]$ asymptotically as 
they are averaged out by integration. The one-dimensional
  marginalised function in $x_1$
has a striking resemblance with the causal effect of an intervention $
\Do(X_{1}=x_1)$ on $Y$. In fact, as in the derivation of
\eqref{backdoor2},
\begin{eqnarray}
& &\E[Y\,\vert\, \Do(X_{1}=x_1)] =\int y p(y\,\vert\, \Do(X_{1}=x_1))dy\nonumber \\
&=&\int \int y p(y\,\vert\, X_{1}=x_1, X_{1}^\Superset)dP( X_{1}^\Superset) dy \label{eq:insertbackdoor}\\
&=&\int \E[Y \,\vert\, X_{1}=x_1, X_{1}^\Superset]dP( X_{1}^\Superset)
    =\E[m(x_1, X_2)]\label{eq:interchange},
\end{eqnarray}
by inserting Pearl's backdoor adjustment formula \eqref{eq:backdoor} in equation
\eqref{eq:insertbackdoor}, assuming that we can interchange the order of
integration, and choosing $X_2 = X_1^\Superset$ in equation
\eqref{eq:interchange}. 
Hence, the causal 
effect $\E[Y\,\vert\, \Do(X_{1}=x_1)]$ can be estimated by marginal
integration of $m(x, X_{1}^\Superset)$.

\subsection{The setting without instantaneous effects and the estimator}
In a time series context, the theory of marginal integration needs to be
extended in order for the same asymptotic results to hold. We consider the
following set-up. 
The data is a finite realization of a strictly stationary and strongly
mixing \citep{fan2005} multivariate Markovian process $(\mathbf{X}_t)_{t
  \in \Z}$ of order $p_0$  and with $l$ 
components, that is, 
$\mathbf{X}_t=(X_{1,t},\ldots,X_{l,t})$ for every $t \in \Z$. The number of
components $l$ is arbitrary (but fixed) and includes the univariate
case. The sequence can display serial dependence between variables
within the same and also between different components. We also assume that the stochastic process can be represented in the
  form of a structural equation model (SEM) which remains invariant across time
  $t$. The SEM consists of a system of equations $\{X_{c,t}=f_c(X_{\pa(c,t)},\epsilon_{c,t});\ c=1,\ldots ,l;\ t \in \Z\}$ in which $\pa(c,t)$ denotes the parental set or the set of direct causes of $(c,t)$ and $\epsilon_{c,t}$ denotes the noise term. We then require that the
  Markovian process has the following property:   
\begin{eqnarray}\label{SEM}
& &{\cal L}(X_{c,t}\mid\mathbf{X}_{t-1},\mathbf{X}_{t-2},\ldots,\mathbf{X}_{t-p_0}) =
  {\cal L}(X_{c,t}\mid X_{\pa(c,t)}),\nonumber\\
&  &\{X_{c,t}\mid X_{\pa(c,t)};\ c=1,\ldots ,l;\ t \in \Z\}\ \mbox{are
     conditionally jointly independent}.
\end{eqnarray}  
We assume here that there are no instantaneous effects, that is $\pa(c,t)
\subseteq \{(d,s);\ d=1,\ldots ,l;\ s=t-1,\ldots ,t-p_0\}$.
According to the SEM, one can construct an infinite directed acyclic graph (DAG). The random variables $\{X_{c,t};\ c=1,\ldots ,l;\ t \in \Z \}$
  correspond to nodes in the DAG and the edges are drawn from each variable in a parental set to its effects. Due to
  stationarity, the DAG does not change over time $t$ and due to the
  Markovian structure, it is sufficient to represent it by a DAG $D$ for
  the time points  $t, t-1,\ldots ,t-p_0$. In terms of the
graph $D$, no instantaneous effects means that the edges are directed forward in time, and there
are no directed edges across the different components at the same time
point. An example is given in Figure \ref{fig:ex_graph}. We will relax the assumption of having no instantaneous effects in
Section \ref{sec:instant}. 

\begin{figure}[!htb]
    \centering
    \includegraphics[width=0.65\textwidth]{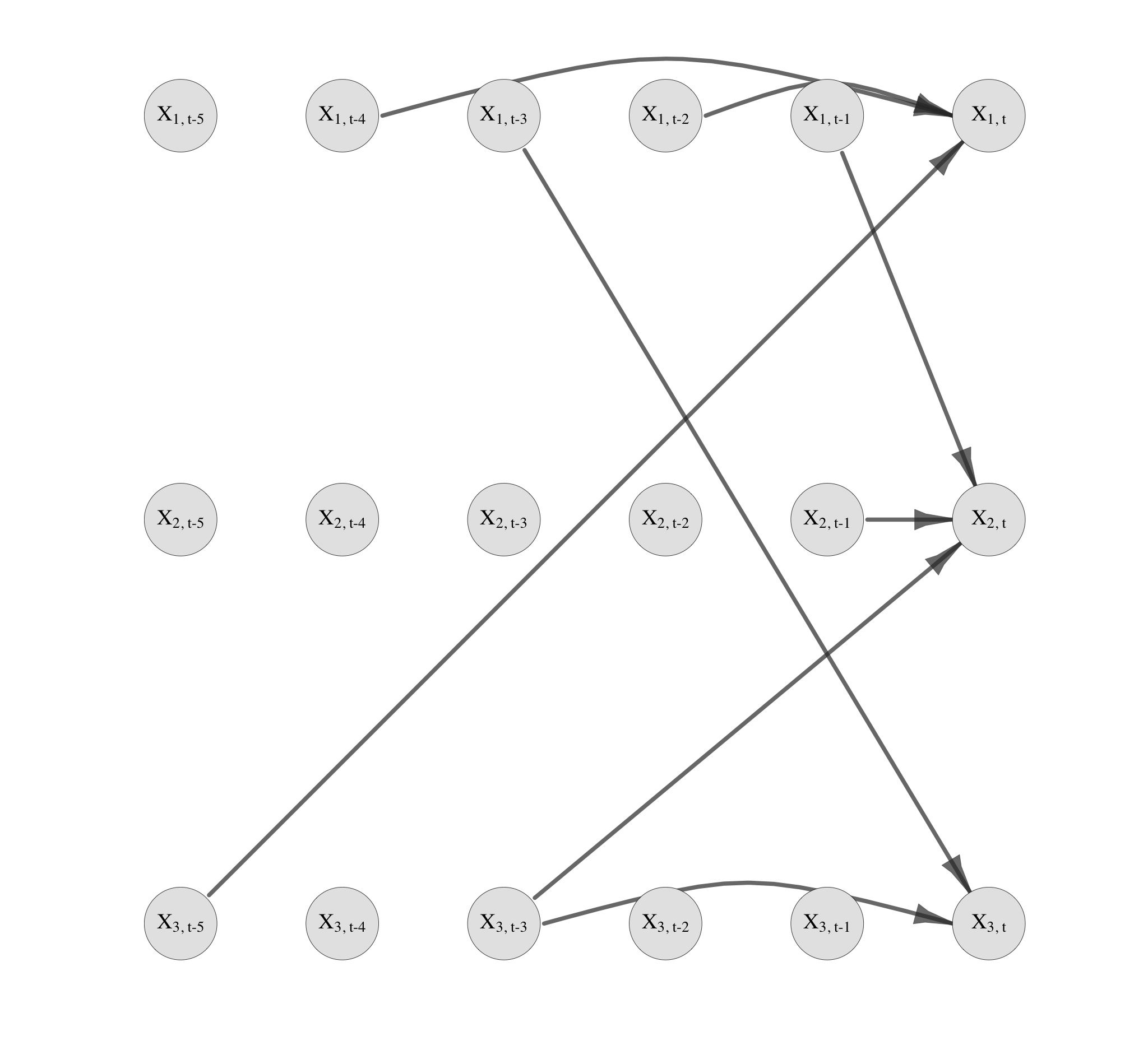}
    \caption{A time-invariant DAG without instantaneous effects which corresponds to the SEM:
   $X_{1,t}=0.8 X_{1,t-2}-0.3X_{1,t-4}+0.5X_{3,t-5}+\epsilon_{1,t}$, $X_{2,t}= -0.7X_{1,t-1}+\sqrt {\vert X_{2,t-1}+X_{3,t-3}\vert}+\epsilon_{2,t}$ and $X_{3,t}= \sin(X_{1,t-3}X_{3,t-3})+\epsilon_{3,t}$. The nodes correspond to  variables $\{X_{c,t};\ c=1,\ldots ,l;\ t \in \Z \}$. Due to stationarity and the
  Markovian structure, it is sufficient to represent the DAG for time points $t,\ldots ,t-p_0$.}
    \label{fig:ex_graph}
\end{figure}

We denote the causal effect
at $X_{c_1,t}$ after an intervention at $X_{c_2,t-s},\, s \in \N$ by
$\E[X_{c_1,t}\,\vert\, \Do(X_{c_2,t-s}$\\$=x)]$, where $c_1,\, c_2$ are
components of the multivariate time series. Also, $t$ is the time index of
the response variable and $s$ the time difference between intervention
and response variable. Let $\mathbf{X}_1,\ldots,\mathbf{X}_n$ be a sample of
the sequence. We estimate the causal effect as  
\begin{equation}\label{eq:mint}
\hat{\E}[X_{c_1,t}\,\vert\, \Do(X_{c_2,t-s}=x)]=(n-s-p)^{-1}\sum_{k=s+p+1}^{n}\hat{m}(x,\mathbf{X}_{k-s}^{\Superset}), 
\end{equation}
where $\mathbf{X}_{t}^{\Superset}$ denotes the set
$\{\mathbf{X}_{t-1},\ldots,\mathbf{X}_{t-p}\}$ with $p \ge p_0$ chosen
  reasonably large such that it is larger than the Markovian order $p_0$.
The partially locally linear estimator $\hat{m}(x,\mathbf{x}^{\Superset})$ is obtained by minimising
\begin{equation} \label{eq:hatm}
\sum_{k=s+p+1}^n (X_{c_1,k}-\alpha-\beta(X_{c_2,k-s}-x))^2 K_{h_1}(X_{c_2,k-s}-x)L_{h_2}(\mathbf{X}_{k-s}^{\Superset}-\mathbf{x}^{\Superset})
\end{equation}
with respect to $\alpha$ and $\beta$. We then use
$\hat{m}(x,\mathbf{x}^{\Superset}) = \hat{\alpha}$. $K$, $L$ are two kernel functions
and $h_1$, $h_2$ are their corresponding bandwidths. When the bandwidths are
chosen appropriately, we are able to show that our estimator recovers the
true causal effect consistently and with convergence rate $n^{-2/5}$. For
this purpose, we require some assumptions outlined next.

\subsection{Asymptotic result}\label{sec:asy_theory}
We assume the following conditions.
\begin{assumption}\label{ass1}\hspace{0.5cm}
	\begin{enumerate}
		\item The mixing coefficients of the underlying strongly
                  mixing stationary $l$-dimensional Markovian process of order $p_0$ satisfy 
		$\alpha_k \leq A k^{-\beta}$ and $\beta > 2+\frac{pl+1}{\gamma}+pl$ for some constants $A>0$, $\gamma>0$ and $p\geq p_0$.
		\item The variables $X_{c,t}$ have bounded support for every $t \in \Z$ and $c \in \{1,\ldots,l\}$.
		\item The variables $X_{c_2,t},\mathbf{X}_{t}^{\Superset}$ have a joint density with respect to Lebesgue measure and $p(u,\mathbf{u}^{\Superset})$ has continuous, bounded partial derivatives up to order 2 with respect to $u$ and up to order $d$ with respect to $\mathbf{u}^{\Superset}$. In addition, for a $\delta>0$ the joint distribution is bounded away from zero in a neighbourhood of $x$ $$\underset{\mathbf{u}^{\Superset} \in \supp(\mathbf{X}_{t}^{\Superset})}{\underset{u \in x\pm \delta}{\inf}} p(u,\mathbf{u}^{\Superset})>0.$$
		\item  The conditional density $p(\mathbf{X}_i^{\Superset}\,\vert\, \mathbf{X}_j,\mathbf{X}_j^{\Superset},\mathbf{X}_k,\mathbf{X}_k^{\Superset},\mathbf{X}_l,\mathbf{X}_l^{\Superset})$ is bounded a.s. for every $i,j,k,l \in \Z$.
		\item For every $j\in \Z$ the joint density $p(\mathbf{X}_{t},\mathbf{X}_{t}^{\Superset},\mathbf{X}_{t+j},\mathbf{X}_{t+j}^{\Superset})$ is bounded.
		\item The regression function $m(u,\mathbf{u}^{\Superset})=\E{[X_{c_1,t}\,\vert\, X_{c_2,t-s}=u, \mathbf{X}_{t-s}^{\Superset}=\mathbf{u}^{\Superset}]}$ exists and has bounded partial derivatives up to order 2 with respect to $u$ and up to order $d$ with respect to $\mathbf{u}^{\Superset}$. Furthermore, 
		$\E{[\vert X_{c_1,t}\vert \, \vert\, X_{c_2,t-s}=u, \mathbf{X}_{t-s}^{\mathcal{S}}=\mathbf{u}^{\Superset}]}$ is bounded.
		\item The kernel functions $K,L$ are symmetric, bounded on a bounded support and $L$ is an order $d$-kernel. 
		\item The product kernel $K\cdot L$ is Lipschitz, i.e., there exists a constant $\Lambda \geq 0$ such that for all $(u,v)$ and $(u',v')$ $$\vert K(u)L(v) -K(u')L(v')\vert \leq \Lambda \Vert(u,v)-(u',v')\Vert\,.$$
		\item The bandwidths are chosen such that $nh_1h_2^{2pl}/log^2(n)\rightarrow \infty$ and $h_1^4 log(n)/ h_2^{pl}\rightarrow 0$, $h_2^d/h_1^2 \rightarrow 0$, $n^{\theta}h_1h_2^{pl}/log(n)\rightarrow \infty$ with $\theta=\frac{\beta-2-pl-(pl+1)/ \gamma}{\beta+2-pl}$.
	\end{enumerate}
\end{assumption}

\hspace{-0.65cm} Assumptions 1.2 - 1.3, 1.6 and the bandwidth conditions in 1.9 are adapted from Assumption~1 in \cite{ernest2015}. Assumptions 1.1, 1.5, 1.7 - 1.9 ensure the uniform convergence of the kernel density estimator for dependent variables. In particular, a mixing rate is defined in Assumption~1.1 which yields the desired rate of convergence of the kernel density estimator.

\begin{theorem}\label{thm1}
	Let $(\mathbf{X}_t)_{t \in \Z}$ be a strictly stationary, strongly
        mixing Markovian process of order $p_0$, and assume that it
          can be uniquely represented in the form of a structural equation
          model with time invariant DAG as in \eqref{SEM} which exhibits no
          instantaneous effects.
Under Assumption~\ref{ass1}, it holds that 
$$\hat{\E}[X_{c_1,t}\,\vert\, \Do(X_{c_2,t-s}=x)]-\E[X_{c_1,t}\,\vert\, \Do(X_{c_2,t-s}=x)]=O(h_1^2)+O_p(1/\sqrt{nh_1}) .$$
\end{theorem}
\begin{remark} \label{remark:rate}
	The rate $O(n^{-2/5})$ can be obtained by choosing $h_1 \asymp
        n^{-1/5}$ if $pl<d$. This corresponds to the optimal rate of
        convergence for the estimation of one-dimensional twice
        differentiable functions.
\end{remark}

\begin{remark} \label{remark:transf}
\begin{sloppypar}
	Theorem~\ref{thm1} extends to the estimation of $\E[g(X_{c_1,t})\,\vert\, \Do(X_{c_2,t-s}=x)]$ for arbitrary
        real-valued transformations $g(\cdot)$ (see Remark 1
        \cite{ernest2015}). For example, this enables the estimation of
        $\E[X_{c_1,t}^2\,\vert\, \Do(X_{c_2,t-s}=x)]$, $\E[\vert X_{c_1,t}\vert\,\vert\, \Do(X_{c_2,t-s}=x)]$ or $P(X_{c_1,t}
        \leq b\,\vert\, \Do(X_{c_2,t-s}=x))$.
This is especially useful in the analysis of financial time series. Financial returns often show no evidence of serial correlation, whereas the absolute returns or the squared returns do.
\end{sloppypar}
\end{remark}

We defer the proof of Theorem \ref{thm1} to \ref{proofthm1}. The proof is a non-trivial extension of the
    techniques in \citet{fan1998} to the case of stationary
    Markovian processes. Alternatively, Theorem \ref{thm1} might be
    derived by generalizing the theory of projections in
    \citet{masry1997} to a fully nonparametric model with more than two
    projection components. 

Furthermore, \citet{hengartner2005} propose an interesting modification
to the marginal integration estimator presented here. They
  suggest an internally normalised pilot estimator for the conditional
  mean which leads to nicer asymptotic results. Related to their
    motivation, we will discuss ways to ease the bandwidth choice and
  to bypass the
  use of higher order kernels for MINT-T in practice in the upcoming Sections
  \ref{section:implementation} and \ref{section:choiceh}. 
  
\subsection{Instantaneous effects} \label{sec:instant}

Up till now, we assumed that there are no instantaneous effects between the
different components of the time series. Here, we will argue that
 some statements are still possible in presence of instantaneous effects
while requiring no knowledge of the underlying DAG in the structural equation model (SEM). Instead of (6) we consider a SEM with DAG $D$ 
\begin{eqnarray}\label{instSEM}
& &X_{c,t} = f_{c}(X_{\pa(c,t)},\epsilon_{c,t}),\ c=1,\ldots,l,\ t \in \Z,\nonumber\\
& &\epsilon_{1,t},\ldots ,\epsilon_{l,t}\ \mbox{jointly independent for all}\ t \in \Z,
\end{eqnarray}
where $\pa(c,t)$ is also allowed to include indices $(d,t)$ for some $d \neq c$, i.e., so-called instantaneous effects (in Figure 1, there would be some directed edges among the variables $X_{1,t},\ldots, X_{l,t}$); and $\pa(c,t)$ is defined with respect to the DAG $D$. 

To deal with instantaneous effects, we have to choose another adjustment
set. Ideally, when doing an intervention $\Do(X_{c_2,t-s} = x)$ we would
choose an adjustment set like $\pa(c_2,t-s)$ or a slightly larger set
containing only ancestors of $(c_2,t-s)$. In absence of knowing the true
underlying DAG, we cannot do this. We propose the following adjustment set
\begin{eqnarray}\label{adj-larger}
X_{c_2,t-s}^{\tilde{\Superset}}:=\cup_{c \neq c_2}\{X_{c,t-s}\} \cup
  \{\mathbf{X}_{t-s-1},\ldots ,\mathbf{X}_{t-s-p}\}
\end{eqnarray}
where $p \geq p_0$. That is, we also include all instantaneous variables
$X_{c,t-s}\ (c \neq c_2)$. The estimator is the same as in~\eqref{eq:mint} but now
using the larger adjustment set in \eqref{adj-larger}. Using the theory as
presented in Section \ref{sec:asy_theory}, the estimator will converge, with rate $n^{-2/5}$,
to 
\begin{eqnarray*}
\tilde{\E}[X_{c_1,t}\,\vert\,\Do(X_{c_2,t-s}=x)] := \int
  \E[X_{c_1,t}\,\vert\,X_{c_2,t-s}=x,X_{c_2,t-s}^{\tilde{\Superset}}]
  dP(X_{c_2,t-s}^{\tilde{\Superset}}). 
\end{eqnarray*}
In general, it will happen that 
\begin{eqnarray*}
\tilde{\E}[X_{c_1,t}\,\vert\,\Do(X_{c_2,t-s}=x)] \neq
  \E[X_{c_1,t}\,\vert\,\Do(X_{c_2,t-s}=x)]. 
\end{eqnarray*}
What we will argue though is that when
$\tilde{\E}[X_{c_1,t}\,\vert\,\Do(X_{c_2,t-s}=x)] \neq \E[X_{c_1,t}]$ (and
its estimate being sufficiently far away from the mean), one can claim a total
causal effect of $X_{c_2,t-s}$ on $X_{c_1,t}$.
For a rigorous statement, we need the following definition. 
\begin{definition}
$X_{c_2,t-s}$ is not total causal for $X_{c_1,t}$ if and only if $X_{c_2,t-s}  \perp X_{c_1,t}$ under the interventional distribution $P^{\mathbf{X}_t\vert \Do(X_{c_2,t-s}=x)}$ for all $x \in \supp(X_{c_2,t-s})$. Here, "$\perp$" denotes independence.
\end{definition}

Apparently, if $X_{c_2,t-s}$ is not total causal for $X_{c_1,t}$, then the
average causal effect $\E[X_{c_1,t}\,\vert\,\Do(X_{c_2,t-s}$\\$=x)]\equiv
\E[X_{c_1,t}]$. We will show next that the same holds true for
$\tilde{\E}[X_{c_1,t}\,\vert\,\Do(X_{c_2,t-s}=x)]$.  

\begin{theorem}\label{th2}
Let $(\mathbf{X}_t)_{t \in \Z}$ be a strictly stationary Markovian process
of order $p_0$, and assume that it is 
        represented in the form of a structural equation model with time
        invariant DAG $D$ as in \eqref{instSEM} allowing also for instantaneous
        effects. Assume that, for all $x \in
        \supp(X_{c_2,t-s})$, the interventional distribution
        $P^{\mathbf{X}_t\vert \Do(X_{c_2,t-s}=x)}$ after the intervention
        $\Do(X_{c_2,t-s}=x)$ is faithful with respect to the interventional
        DAG $D_{\mathrm{interv}}$ where all directed arrows into
        $(c_2,t-s)$  are deleted. Then from
       $\tilde{\E}[X_{c_1,t}\,\vert\,\Do(X_{c_2,t-s}=x)] \neq
          \E[X_{c_1,t}]$ for some $x \in \supp(X_{c_2,t})$, it follows that $X_{c_2,t-s}$ is total causal for $X_{c_1,t}$ for $s>0$. 
\end{theorem}
\begin{remark} The distribution
of the stochastic process does not necessarily define a unique DAG. We note that the statement is true for any DAG $D$ such that the faithfulness
and the Markov property hold.
\end{remark}
\begin{remark} If we know beforehand that we are only interested in a set
  of prespecified intervention values $\mathcal{I}\subset
  \supp(X_{c_2,t-s})$, then the above assumption  "$P^{\mathbf{X}_t\vert
    \Do(X_{c_2,t-s}=x)}$ is faithful with respect to $D_{\mathrm{interv}}$
  for all $x \in 
        \supp(X_{c_2,t-s})$" can be weakened to require only
"for all $x \in \mathcal{I}$". The statement of the theorem then
  reads: from
       $\tilde{\E}[X_{c_1,t}\,\vert\,\Do(X_{c_2,t-s}=x)] \neq
          \E[X_{c_1,t}]$ for some $x \in \mathcal{I}$, it follows that 
$X_{c_2,t-s}$ is total causal for $X_{c_1,t}$ for $s>0$.
\end{remark} 
 
Just as in Theorem \ref{thm1}, we are able to estimate
$\tilde{\E}[X_{c_1,t}\,\vert\,\Do(X_{c_2,t-s}=x)]$ 
with the optimal rate by adjusting on $X_{c_2,t-s}^{\tilde{\Superset}}$
from \eqref{adj-larger} in the presence of instantaneous effects. Thus,
  if $\widehat{\tilde{\E}}[X_{c_1,t}\,\vert\,\Do(X_{c_2,t-s}=x)] \neq
    \overline{X}_{c_1,.}$,
we would
  claim a total causal effect and hence avoid false positive statements
  about causal effects. A word of caution should be added, however:  we
  would typically rank different causal effects by a quantity like
\begin{eqnarray*}
\mathrm{eff}\left(X_{c_2,t-s} \to X_{c_1,t}\right) =  \int
  \|\E[X_{c_1,t}\,\vert\,\Do(X_{c_2,t-s}=x)] - \E[X_{c_1,t}]\| dw(x) 
\end{eqnarray*}
for some weight function $w(\cdot)$. This causal effect quantity can be
rather different from 
\begin{eqnarray*}
\widetilde{\mathrm{eff}}\left(X_{c_2,t-s} \to X_{c_1,t}\right) =  \int
  \|\tilde{\E}[X_{c_1,t}\,\vert\,\Do(X_{c_2,t-s}=x)] - \E[X_{c_1,t}]\| dw(x). 
\end{eqnarray*}
Thus, a ranking of total causal effects by estimates of
$\widetilde{\mathrm{eff}}(\cdot)$ can be rather different than by estimates
of the true causal effects $\mathrm{eff}(\cdot)$.


\begin{proof}[Proof of Theorem \ref{th2}]
To  simplify notation and without loss of generality assume that $\E[X_{c,t}]=0$ for all $c$.
We will show the reverse implication: $X_{c_2,t-s}$ is not total causal for $X_{c_1,t} \Rightarrow \tilde{\E}[X_{c_1,t}\,\vert\,\Do(X_{c_2,t-s}$\\$=x)]\equiv 0$. If $X_{c_2,t-s}$ is not total causal for $X_{c_1,t}$, $X_{c_1,t}$ and $X_{c_2,t-s}$ are independent under the interventional distribution
$P^{\mathbf{X}_t\vert \Do(X_{c_2,t-s}=x)}$ for all $x \in \supp(X_{c_2,t-s}) $.
On a graphical level, the intervention $\Do(X_{c_2,t-s}=x)$ corresponds to deleting all
incoming edges into $(c_2,t-s)$ from the observational DAG $D$ and substituting $X_{c_2,t-s}$ by $x$ in
the resulting interventional DAG $D_{\mathrm{interv}}$.
The independence between $X_{c_1,t}$ and $X_{c_2,t-s}$ can
be translated into a graphical criterion, namely that $X_{c_1,t}$ and
$X_{c_2,t-s}$ are d-seperated by the empty set in $D_{\mathrm{interv}}$ to which $P^{\mathbf{X}_t\vert \Do(X_{c_2,t-s}=x)}$ is faithful for all $x$. Then, there can only be paths between $X_{c_2,t-s}$ and
$X_{c_1,t}$ that contain at least one v-structure. We denote such a v-structure by
$X_{c_{\mathrm{source}_1},t-j_1}\rightarrow
X_{c_{\mathrm{collider}},t-j_2}\leftarrow
X_{c_{\mathrm{source_2}},t-j_3}$. 
Here, $c_{\mathrm{source}_1}$ or $c_{\mathrm{source}_2}$ could be equal to $c_{\mathrm{collider}}$ but not both.

If the collider
$X_{c_{\mathrm{collider}},t-j_2}$ lies between $j_2=0,..,s-1$, nothing
changes by adjusting on $X_{c_2,t-s}^{\tilde{\Superset}}$. The path still
remains blocked  by $X_{c_{\mathrm{collider}},t-j_2}$. On the other hand,
if $j_2=s$, then  $j_1, j_3\geq s$ and the
path is blocked by either $X_{c_{\mathrm{source}_1},t-j_1}$ or
$X_{c_{\mathrm{source}_2},t-j_3}$ after adjusting on
$X_{c_2,t-s}^{\tilde{\Superset}}$. Therefore, every path between
$X_{c_1,t}$ and $X_{c_2,t-s}$ is blocked by
$X_{c_2,t-s}^{\tilde{\Superset}}$ in $D_{\mathrm{interv}}$. 

Since $P$ is Markov w.r.t. $D$, we have conditional independence of
  $X_{c_1,t}$ and $X_{c_2,t-s}$ given $X_{c_2,t-s}^{\tilde{\Superset}}$ and therefore, 
\begin{eqnarray*}
& &\tilde{\E}[X_{c_1,t}\,\vert\,\Do(X_{c_2,t-s}=x)]=
\int \E[X_{c_1,t}\,\vert\,X_{c_2,t-s}=x,X_{c_2,t-s}^{\tilde{\Superset}}]
  dP(X_{c_2,t-s}^{\tilde{\Superset}})\\
&=&\int \E[X_{c_1,t}\,\vert\,X_{c_2,t-s}^{\tilde{\Superset}}]
  dP(X_{c_2,t-s}^{\tilde{\Superset}}) = \E[X_{c_1,t}] = 0.
\end{eqnarray*}

\end{proof}

\section{Implementation}\label{section:implementation}

The estimator in \eqref{eq:mint} is constructed using the partially locally
linear estimator in \eqref{eq:hatm}. This requires choosing two bandwidths
$h_1$ and $h_2$, and such a choice is not easy in view of the fact that we
cannot rely on cross-validation for the quantity $\E[X_{c_1,t}\,\vert\,
\Do(X_{c_2,t-s}=x)]$ (since there is no corresponding loss to an
observable quantity). The related estimator suggested by
\cite{ernest2015}, based on a boosting idea, seems substantially easier for
practical purposes.  

\subsection{MINT-T: an implementation for marginal integration}

We describe here our estimation
  scheme based on the boosting idea in \cite{ernest2015}: we call it
  "\mbox{MINT-T}'', standing for marginal integration in time series. \mbox{MINT-T} differs from the
  estimation scheme in \cite{ernest2015} in that no additive approximation
  is used in the first step. The more regular structure of our setting
  (e.g., stationarity assumption, same size of adjustment set for all
  interventions) allows us to directly apply the marginal integration estimator
  from the first step. 

Exploiting the strict stationarity of the time series, we obtain $n-s-p$
dependent samples (i.e., samples of "regressors''). For $X_{c_1,t}$ and
$X_{c_2,t-s}$, we have the "response'' vector  $\mathbb{X}_{c_1} := \{X_{c_1,n},
X_{c_1,n-1},\ldots,X_{c_1,s+p+1}\}$ and the "regressor'' $\mathbb{X}_{c_2} :=
\{X_{c_2,n-s}, X_{c_2,n-s-1}, \ldots, X_{c_2,p+1}\}$, respectively. For the adjustment set $\textbf{X}^\mathcal{S}_{t-s}$, we have the lagged values of
  the "regressors'' which can be represented by the matrix
  $(\textbf{X}_{n-s-1}(p),\ldots ,\textbf{X}_{p}(p))$, where $\textbf{X}_{t}(p) = (\textbf{X}_t,\ldots
  ,\textbf{X}_{t-p+1})$. This matrix involves the samples
  $\mathbb{X}^{\mathcal{S}}:=\{\textbf{X}_{n-s-1}, \ldots ,
  \textbf{X}_1\}$.

The initial step of \mbox{MINT-T} consists of approximating the
regression function $m(x, \textbf{x}^{\mathcal{S}}) =
\E[X_{c_1,t} \mid X_{c_2,t-s} = x, \textbf{X}^{\mathcal{S}}_{t-s} =
\textbf{x}^{\mathcal{S}}]$ by a locally constant estimator of the form
\begin{equation} \label{eq:locconst}
	\hat{m}_{\text{init}}(x, \textbf{x}^{\mathcal{S}}) := \argminn\limits_{\alpha} \sum\limits_{k=s+p+1}^n (X_{c_1,k} - \alpha)^2 K_{h_1}(X_{c_2, k-s} - x) L_{h_2}(\textbf{X}^{\mathcal{S}}_{k-s} - \textbf{x}^{\mathcal{S}}) .
\end{equation}

Marginally integrating the estimator~\eqref{eq:locconst} over the samples $\mathbb{X}^{\mathcal{S}}$ with the empirical mean as in equation~\eqref{eq:hatm} yields an estimate for the true causal effect $\E[X_{c_1,t}\,\vert\,\Do(X_{c_2,t-s}=x)]$. 
The problem is that the marginally integrated estimator~\eqref{eq:hatm} is very sensitive to the choice of the bandwidths $h_1$ and $h_2$. Moreover, we cannot use cross-validation or penalised likelihood techniques to determine the optimal bandwidths as $\E[X_{c_1,t}\,\vert\,\Do(X_{c_2,t-s}=x)]$ is neither a regression function nor does it appear in the likelihood. To make our estimator more robust with respect to the choice of the bandwidths, we therefore apply $B$ steps of $L_2$-boosting with the locally constant estimator~\eqref{eq:locconst}, which, in every iteration, is applied to the residuals of the previous fit. The key idea of the boosting procedure is that the bandwidths $h_1$ and $h_2$ in~\eqref{eq:locconst} can be set to large values in order to obtain an estimator with high bias and low variance. The boosting iterations then reduce the bias. As such, the boosted estimator is less sensitive to the specific choice of the bandwidths as long as they are sufficiently large. This will be shown experimentally in Section~\ref{section:choiceh}.
 The effect of the boosting can be compared to the one of the use of a higher-order kernel~\citep{dimarzio2008}.

We now describe the boosting procedure in detail. Let
$\hat{m}_1 := \hat{m}_{\text{init}}$ defined in~\eqref{eq:locconst}. Then, the $n-s-p$ residuals
$R_{1,s+p+1},\ldots,R_{1,n}$ of the initial model fit are given as 
 $$
 	R_{1,k} = X_{c_1,k} - \hat{m}_{1}(X_{c_2, k-s}, \textbf{X}^{\mathcal{S}}_{k-s}), \qquad k=s+p+1,\ldots,n.
 $$ 
 The locally constant fit of the residuals is then obtained as in~\eqref{eq:locconst} by minimising
 \begin{equation} \label{eq:mint_kernel}
 	\sum\limits_{k=s+p+1}^n (R_{1,k} - \alpha)^2 K_{h_1}(X_{c_2, k-s} - x) L_{h_2}(\textbf{X}^{\mathcal{S}}_{k-s} - \textbf{x}^{\mathcal{S}})
 \end{equation}
 with respect to $\alpha$, and is denoted by $\hat{g}_{R_1}(x, \textbf{x}^{\mathcal{S}}):= \hat{\alpha}$.
 Let $\hat{\textbf{m}}$ be the $(n-s-p)$-dimensional vector of $\hat{m}$ evaluated at the samples of the time series and $\textbf{X}_{c_1}$ be the $(n-s-p)$-dimensional vector of the samples in $\mathbb{X}_{c_1}$. We can then summarise the $L_2$-boosting step as follows: for $b=1,\ldots,B-1$,
\begin{align*}
	\hat{m}_{b+1} & = \hat{m}_b + \hat{g}_{R_{b}}, \\
	\textbf{R}_{b+1} & = \textbf{X}_{c_1} - \hat{\textbf{m}}_{b+1},
\end{align*}
where $B$ (the number of boosting iterations) is a regularisation parameter. 

Finally, we marginally integrate over the samples
$\mathbb{X}^{\mathcal{S}}$ with the empirical mean. This last step of \mbox{MINT-T}
yields the final estimate  
$$
	\hat{\E}[X_{c_1,t} \mid \Do(X_{c_2,t-s} = x)] = (n-s-p)^{-1} \sum\limits_{k=s+p+1}^{n} \hat{m}_B (x, \textbf{X}_{k-s}^{\mathcal{S}}) .
$$
The pseudo-code summarising our method is provided in Algorithm~\ref{alg:Mint}.

\begin{algorithm}[!htb]
\begin{algorithmic}[1]
\vspace{0.4cm}
\STATE Construct $n-s-p$ samples of the adjustment set, the intervention variable and the target variable exploiting the strict stationarity of the time series.
\STATE Fit an initial locally constant estimator of $X_{c_1,t}$ versus $X_{c_2,t}$ and $\textbf{X}^{\mathcal{S}}_{t-s}$ with large bandwidths to obtain $\hat{m}_{1} := \hat{m}_{\text{init}}$ in~\eqref{eq:locconst}.
\FOR{$b=1,\ldots,B-1$}
\STATE Apply one step of $L_2$-boosting as follows:  
\STATE (i)  \hspace{0.2cm} Compute residuals $\textbf{R}_{b} = \textbf{X}_{c_1} - \hat{\textbf{m}}_{b}$
\STATE (ii) \hspace{0.14cm} Fit the residuals with the kernel estimator \eqref{eq:mint_kernel} to obtain $\hat{g}_{R_b}$  
\STATE (iii) \hspace{0.03cm} Set $\hat{m}_{b+1} = \hat{m}_b + \hat{g}_{R_b}$
\ENDFOR
\RETURN Do marginal integration: output $(n-s-p)^{-1} \sum\limits_{k=s+p+1}^n
  \hat{m}_{B}(x,\textbf{X}^{\mathcal{S}}_{k-s})$
\end{algorithmic}
\caption{MINT-T}\label{alg:Mint}
\end{algorithm}

\section{Empirical results}\label{section:empirical}

We provide here empirical results of the marginal integration
  method \mbox{MINT-T} for the estimation of causal effects. We also compare it to
  a reference method, explained below, which
  relies on approximating the data-generating stochastic process. This is
  of course a very ambitious task and, in its full generality, exposed to the curse of
  dimensionality. 
  
\subsection{A reference method}
For comparison, we consider a reference method where we assume that the
time series has an additive functional form with an additive Gaussian error
term. This 
  assumption may easily fail though and thus, the method is exposed to
  model misspecification. We then estimate the value of
each of the $l$ components of $\textbf{X}_t$ by an additive function of the
$p$ previous values of all components, that is, 
\begin{align}
	X_{1, t} & = \hat{\mu}^{(1)} + \sum\limits_{c=1}^l \sum\limits_{j=1}^p \hat{m}^{(1)}_{c, t-j}(X_{c, t-j}) + \hat{\epsilon}_{1,t} \nonumber \\
	& \ \  \vdots \label{eq:referenceMethod} \\
	X_{l, t} & = \hat{\mu}^{(l)} + \sum\limits_{c=1}^l \sum\limits_{j=1}^p \hat{m}^{(l)}_{c, t-j}(X_{c, t-j}) + \hat{\epsilon}_{l,t} \nonumber .
\end{align}
We now set $X_{1,j} = X_{2,j} =,\ldots,= X_{l,j} = 0$ for $j=1,\ldots,p$ and then
iteratively simulate the subsequent values at time points $p+1,p+2,\ldots,n$
of the time series using the estimated functions and estimated error terms
from the additive model~\eqref{eq:referenceMethod} with one exception: when
reaching time point $n-s$, we intervene on component $c_2$ by setting
$X_{c_2, n-s}$ to the value $x$. In the end, we record the simulated value
$\hat{X}_{c_1,n}$. We repeat the whole procedure $N$ times to obtain $N$
simulated realizations $\{\hat{X}_{c_1,n}^{(1)},\ldots,\hat{X}_{c_1,n}^{(N)}\}$. For
sufficiently large $N$, the total causal effect at  $X_{c_1,t}$ after an
intervention at $X_{c_2, t-s}$ can be estimated as 
$$
	\hat{\E}[X_{c_1,t} \mid \Do(X_{c_2, t-s} = x)] = N^{-1} \sum\limits_{i=1}^N \hat{X}_{c_1,n}^{(i)} .
$$

\subsubsection{Approximating the true causal effect}

If the functional form of the true underlying time series and the
distributions of the error terms are known, we can use the reference method
for computing the true causal effect. We then simply replace the estimated
functions and noise variables in the additive
model~\eqref{eq:referenceMethod} by the (not necessarily additive) true
ones, but apart from that stick to the simulation procedure described
above.

\subsection{Simulations}\label{section:simulations}

We examine here \mbox{MINT-T} on simulated time series from a variety of models
covering linear to nonlinear, additive to non-additive, and univariate to
multivariate models: 
\begin{itemize}
\item Model 1: \hspace{0.5cm}$X_t=0.4X_{t-2}-0.6X_{t-6}+0.3X_{t-10}+\epsilon_t$
\item Model 2: \hspace{0.5cm}$X_t=\cos(X_{t-1}+X_{t-4})+\log(\vert X_{t-6}-X_{t-10}\vert+1)+\epsilon_t$
\item Model 3: \hspace{0.5cm}$X_t=\sigma_t\epsilon_t$ with $\sigma_t^2=0.1+0.4X_{t-1}^2+0.2X_{t-4}^2$
\item Model 4: \hspace{0.5cm}$X_t=\sigma_t\epsilon_t$ with $\sigma_t^2=0.2+0.6X_{t-1}^2+0.3\sigma_{t-1}^2$
\item Model 5: \hspace{0.5cm}$X_t=0.4X_{t-1}-0.2X_{t-2}+0.3X_{t-3}+0.8\epsilon_{t-1}+\epsilon_t$
\item Model 6: \hspace{0.5cm}$ \left( \begin{array}{c}
X_{1,t} \\
X_{2,t} \\
X_{3,t} \\
 X_{4,t} \end{array} \right)= \left( \begin{array}{l}
0.4 X_{1,t-1}-0.2X_{1,t-2}+0.3X_{2,t-3}+\epsilon_{1,t}\\
 \cos(X_{1,t-1})+\log(\vert X_{2,t-2}\vert+1)+\epsilon_{2,t}\\
  \sin(X_{3,t-1}-X_{2,t-1})+\sqrt{\vert X_{2,t-3}+X_{4,t-1}\vert}+\epsilon_{3,t}\\
  \cos(X_{2,t-1}-X_{3,t-4})+\log(\vert X_{1,t-6}+X_{2,t-10}\vert+1)+\epsilon_{4,t}
\end{array} \right) $
\end{itemize}
The first model is a linear $\AR(10)$ and model 2 a nonlinear, non-additive
$\AR(10)$-model. The third model corresponds to an $\ARCH(4)$-model and the
fourth model to a $\GARCH(1,1)$-model. The fifth is an $\ARMA(3,2)$-model
and the last a multivariate time series model with four (additive and
non-additive) components. In our simulation study we choose i.i.d. Gaussian
noise with mean zero and variance 1 for models 1-3, 6 and variance 0.5 for
models 4 and 5. In all our numerical experiments, we choose sample size $n =
1000$. 

For each model, we inspect the mean squared error (MSE) between
the true and the estimated causal effect. More precisely, the MSE consists
of the true effect subtracted from the estimated causal effect averaged over
20 cause and effect pairs:  
$$\frac{1}{20}\sum_{j=1}^{20}\sum_{i=1}^9
(\hat{\mathbb{E}}[X_{t}\vert
\Do(X_{t-j}=d_i)]-\mathbb{E}[X_{t}\vert
\Do(X_{t-j}=d_i)])^2.$$

We employ the nine deciles $d_1,\ldots,d_9$ of
the simulated time series (quantiles corresponding to the probabilities
$0.1, 0.2,\ldots,0.9$) as the intervention values. For the
multivariate model 6, we sample the
components of the intervention and response variable uniformly, repeat the calculation over 5 repetitions and average over the resulting MSE values. 
Our method \mbox{MINT-T} requires the following
tuning parameters: the bandwidths $h_1$ and $h_2$, the time lag $p$ and the number of
boosting iterations $B$.  We use the true time lag $p$ whenever it is
known. This is mainly for comparison reason so that the reference method is
not disadvantaged. For the non-Markovian models 4 and 5, we set $p$ to
10. The bandwidth $h_1$ is set to $2\hat{\sigma}$ 
for univariate and $3\hat{\sigma}_{c_2}$ for multivariate time series, where
$\hat{\sigma}$ denotes the empirical standard deviation of the time
series. We used a product Gaussian kernel for $L$ in equation
\eqref{eq:locconst}. The bandwidth $h_2$ is a $p$-dimensional vector for
univariate time series and chosen as $(2\hat{\sigma},\ldots,2\hat{\sigma})$
due to stationarity, while for multivariate time series, $h_2$  is a
$pl$-dimensional vector and each entry is scaled by the standard deviation
of the corresponding component,
i.e., $(3\hat{\sigma}_{1},\ldots,3\hat{\sigma}_{l},\ldots,3\hat{\sigma}_{1},\ldots,3\hat{\sigma}_{l})$. 
In the univariate case, $h_1$ and each element from $h_2$ are equal. For
simplicity and with slight abuse of notation, we refer to both $h_1$ and
each entry of $h_2$ as $h$, where for the multivariate case we actually
mean that the bandwidth is scaled by the corresponding standard deviation
of the component. The number of boosting iterations $B$ is set to
10. The exact choices of $h$, $B$ are not crucial as long as both 
parameters are chosen reasonably large. Moreover, our estimator is
rather insensitive to the choice of $p$. This will be explained in more 
detail in Sections~\ref{section:choiceh} and \ref{section:choicep}. For the
reference method we simulated 1000 and for the true causal effect we
simulated 10000 time series for every intervention
variable and intervention value. 

\begin{sloppypar}The experimental findings are summarised in Table \ref{tab:modelcomp}. The
relative gain is calculated from the formula
$(\text{MSE}_{\mathrm{reference}}-\text{MSE}_{\mathrm{MINT-T}})/\text{MSE}_{\mathrm{reference}}$
and the acceleration factor from
$\text{time}_{\mathrm{reference}}/\text{time}_{\mathrm{MINT-T}}$. 
Our empirical results show that \mbox{MINT-T} outperforms the reference
procedure on all 
models except for the univariate $\AR$-models (model 1 and 2). In theory, we expect the reference to have an advantage when
the underlying model is additive, as it is the case with model 1. Even though model 2 is a nonlinear $\AR$-model, our results show that it can be well approximated through additive modelling. \mbox{MINT-T} is able to provide a
relative gain of 17$\%$-53$\%$ over the reference on the $\ARCH$-, $\GARCH$-, $\ARMA$- and the multivariate model.\end{sloppypar}

In some situations, it is of interest to choose intervention values
  that lie outside of the range of the time series. In
  Table~\ref{tab:modelcomp_interv}, we set the intervention value $d_i$ to
  3 times the $i^{th}$ decile of the simulated time series. As a result,
  the range of the simulated true causal effects becomes larger, and the
  MSE of both methods increases. In comparison, \mbox{MINT-T} remains more robust
  for intervention values that lie outside of the range of the data
  points and achieves a relative gain of at least 45$\%$ on all models.

\begin{table}[h!]
\centering
\begin{tabular}{lrrrrrrr}
\hline
 &\multicolumn{3}{c}{MSE}&\multicolumn{3}{c}{Time$[s]$}&\multicolumn{1}{c}{ True effect} \\
 &MINT-T&Reference&Gain/loss&MINT-T&Reference&Accel.& \multicolumn{1}{c}{between}\\
\hline
Model 1&0.0804& 0.0682 &-17.89\%&16.17&880.15&54 &[-0.6355,0.5645] \\
Model 2&0.0459& 0.0291&-57.73 \%& 16.26&922.38&57&[0.4408,1.2937]\\
Model 3&0.0026 &0.0046&+43.48\%& 8.48&419.74&49&0\\
Model 4 &0.0014 &0.0027&+48.15\% &17.14& 938.65&55&0\\
Model 5& 0.0333 &0.0711 &+53.16\%&17.17&946.76&55&[-0.381, 0.3267] \\
Model 6&0.1430 &0.1724&+17.05\%& 58.18&14788.29&254&[0.3114, 1.9647]\\
\hline
\end{tabular}
\caption{Comparison of \mbox{MINT-T} against the reference in terms of MSE and CPU time consumption per index pair.  Relative gain (indicated by +) and loss (indicated by -), and acceleration factor for the CPU time.}\label{tab:modelcomp}
\end{table}

\begin{table}[h!]
\centering
\begin{tabular}{lrrrrrrr}
\hline
 &\multicolumn{3}{c}{MSE}&\multicolumn{3}{c}{Time$[s]$}&\multicolumn{1}{c}{ True effect} \\
 &MINT-T&Reference&Gain/loss&MINT-T&Reference&Accel.& \multicolumn{1}{c}{between}\\
\hline
Model 1&0.1791& 0.5176&+65.40\%&21.26 &1036.13 &49 &[-1.9433,1.6647] \\
Model 2&0.4688& 1.2139&+61.38 \%& 33.85 &2086.42&62&[ 0.4538,1.9718]\\
Model 3&0.0261& 0.0475&+45.05\%& 11.76 &523.72&45&0\\
Model 4 &0.0079& 0.0213 &+62.91\% &19.87& 1064.73&54&0\\
Model 5& 0.1081& 0.4310 &+74.92\%&21.76&1021.42&47&[-1.1320, 0.9798] \\
Model 6& 0.5919& 2.6364 & +77.55\%& 64.53& 14506.34 &225 &[  0.3090, 3.1312 ]\\
\hline
\end{tabular}
\caption{Comparison of \mbox{MINT-T} against the reference in terms of MSE and CPU time consumption per index pair.  Relative gain (indicated by +) and loss (indicated by -), and acceleration factor for the CPU time. Unlike Table~\ref{tab:modelcomp} and \ref{tab:modelcomp_sq}, the intervention value $d_i$ here is equal to 3 times the $i^{th}$ decile of each time series.}\label{tab:modelcomp_interv}
\end{table}

In many applications, we are interested in the effect of an intervention on
a transformed response variable. For example, any causal effect is
identical to zero in $\ARCH$-models since  
\begin{eqnarray}
\E[X_t\,\vert\,\Do(X_{t-s}=x)]&=&\int \E[X_t\,\vert\,X_{t-s}=x,X_{t-s}^\Superset] dP(X_{t-s}^\Superset) \nonumber \\
&=&\int \E[\sigma_t\epsilon_t\,\vert\,X_{t-s}=x,X_{t-s}^\Superset] dP(X_{t-s}^\Superset) \nonumber \\
&=&\int \E[\sigma_t\,\vert\,X_{t-s}=x,X_{t-s}^\Superset] \E[\epsilon_t]dP(X_{t-s}^\Superset) =0. \nonumber
\end{eqnarray} 
An intervention on a squared response variable is usually nonzero in $\ARCH$-models and corresponds to an intervention  on the volatility function up to a constant:
\begin{eqnarray}
\E[X_t^2\,\vert\,\Do(X_{t-s}=x)]&=&\int \E[X_t^2\,\vert\,X_{t-s}=x,X_{t-s}^\Superset] dP(X_{t-s}^\Superset) \nonumber \\
&=&\int \E[\sigma_t^2\epsilon_t^2\,\vert\,X_{t-s}=x,X_{t-s}^\Superset] dP(X_{t-s}^\Superset) \nonumber \\
&=&\int \E[\sigma_t^2\,\vert\,X_{t-s}=x,X_{t-s}^\Superset]E[\epsilon_t^2]dP(X_{t-s}^\Superset) \nonumber \\
&=&\Var(\epsilon_t)\E[\sigma_t^2\,\vert\,\Do(X_{t-s}=x)] \nonumber.
\end{eqnarray} 

\begin{table}[!htb]
\centering
\begin{tabular}{lrrrrrrr}
\hline
 &\multicolumn{3}{c}{MSE}&\multicolumn{3}{c}{Time$[s]$}&\multicolumn{1}{c}{ True effect} \\
 &MINT-T&Reference&Gain/loss&MINT-T&Reference&Accel.& \multicolumn{1}{c}{between}\\
\hline
Model 1&0.0297 &0.0482 &+38.38\%&17.00&967.50&57&[1.1372,1.5856]\\
Model 2&  0.1647 &0.1746&+5.67\%& 17.49&957.27&55&[1.7211,3.2016]\\
Model 3& 0.0026 &0.0045& +42.22\%& 8.76&401.30&46&[0.1523,0.2924] \\
Model 4 & 0.0008 &0.0009&+11.11\%& 17.12& 930.56&54&[0.0502,0.0940]\\
Model 5 &0.0125& 0.0407& +69.29\%&17.52 &929.60&53&[0.3184,0.4700]  \\
Model 6&1.1216 &1.2232& +8.31 \%&60.51&16051.16&265&[1.5119,5.4317]  \\
\hline
\end{tabular}
\caption{Comparison of \mbox{MINT-T} against the reference in
  terms of MSE and CPU time consumption per index pair for the squared
  response variable. Relative gain (indicated by +) and loss (indicated by -), and acceleration factor for the CPU time. } 
\label{tab:modelcomp_sq}
\end{table}

Similarly, stationary $\GARCH$-processes can be rewritten as $\ARCH(\infty)$ processes, and the causal effect is identical to 0, while a squared $\GARCH$-process can be rewritten as a causal and invertible $\ARMA$-process under certain conditions \cite[Proposition 4.2]{fan2005}.  
Our approach allows for the estimation of causal effects on transformed response variables with arbitrary real-valued transformations (see Remark~\ref{remark:transf}). We repeat the analysis for a  squared response variable, and the experimental results are listed in Table~\ref{tab:modelcomp_sq}. \mbox{MINT-T} outperforms the reference on all models. The relative gain lies between 6$\%$ on the nonlinear $\AR$-model to 69$\%$ on the $\ARMA$-model.

Another advantage of \mbox{MINT-T} is the computation time. The computation time of
\mbox{MINT-T} depends mainly on the size of the adjustment set and the number of
boosting iterations. On the other hand, prediction and fitting of the
generalised additive models require most of the computation time for the
reference method. The reference method fits once and predicts once for
every component at every time point. \mbox{MINT-T} provides an acceleration by a factor
of 53, on average, for univariate and an acceleration by a factor of 248, on average,
for multivariate time series. Therefore, \mbox{MINT-T} remains feasible for
multivariate time series in potentially
large-dimensional problems.

\subsubsection{The choice of the bandwidth $h$}
\label{section:choiceh}

\begin{figure}[!htb]
    \centering
    \begin{subfigure}[b]{0.7\textwidth}
        \includegraphics[width=\textwidth]{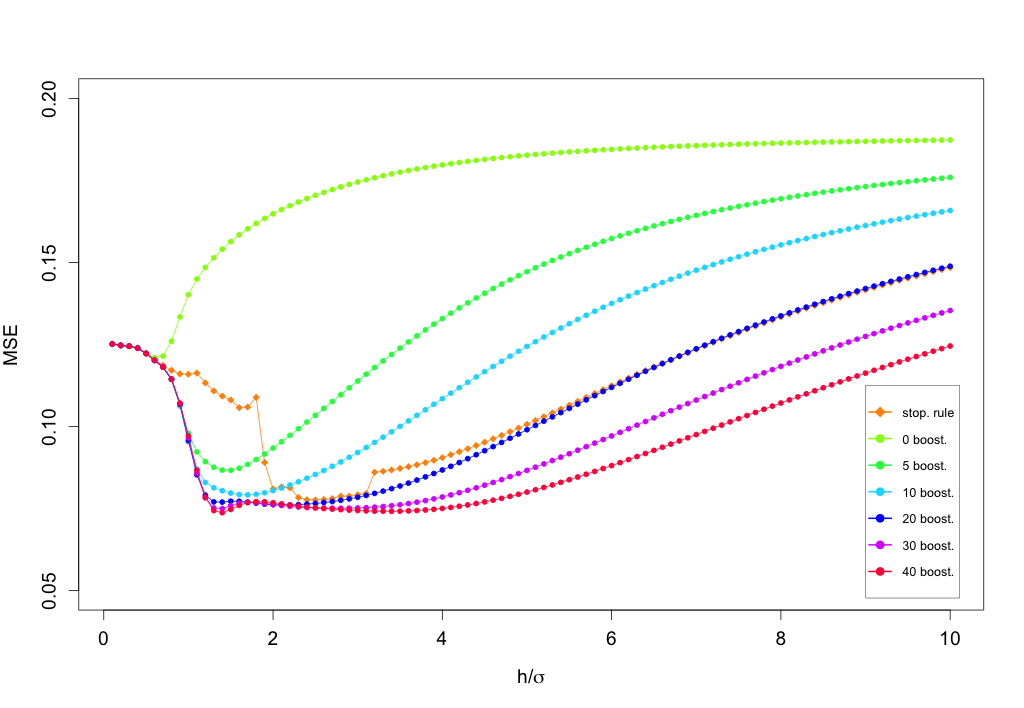}
        \caption{Model 1}
        \label{fig:model1_h}
    \end{subfigure}
    
    \begin{subfigure}[b]{0.7\textwidth}
        \includegraphics[width=\textwidth]{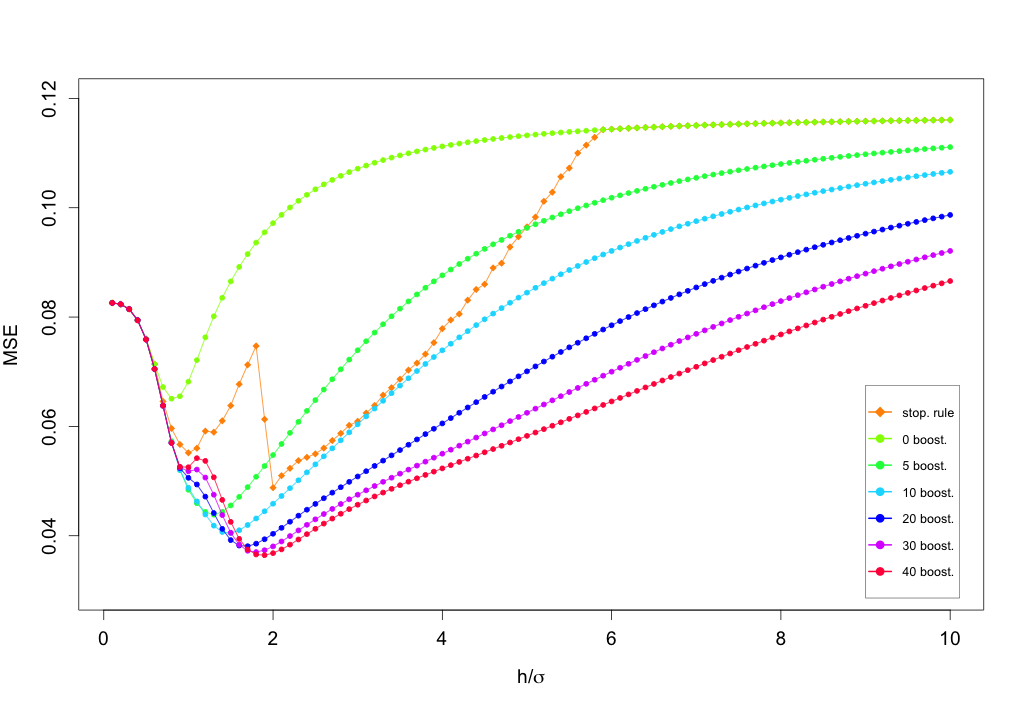}
        \caption{Model 2}
        \label{fig:model2_h}
    \end{subfigure}
\caption{Dependence on bandwidth $h$. MSE values, for models 1 and 2, for $h$
   between $0.1\hat{\sigma}$ 
   and $10\hat{\sigma}$ ($x$-axis with scaled $h/\hat{\sigma}$) for \mbox{MINT-T},
   without boosting, with fixed number of 
   boosting iterations, and with stopping rules. The time lag $p$ is set to 10.}\label{fig:bandwidth1} 
\end{figure}

\begin{figure}[!htb]
    \centering
     \begin{subfigure}[b]{0.7\textwidth}
        \includegraphics[width=\textwidth]{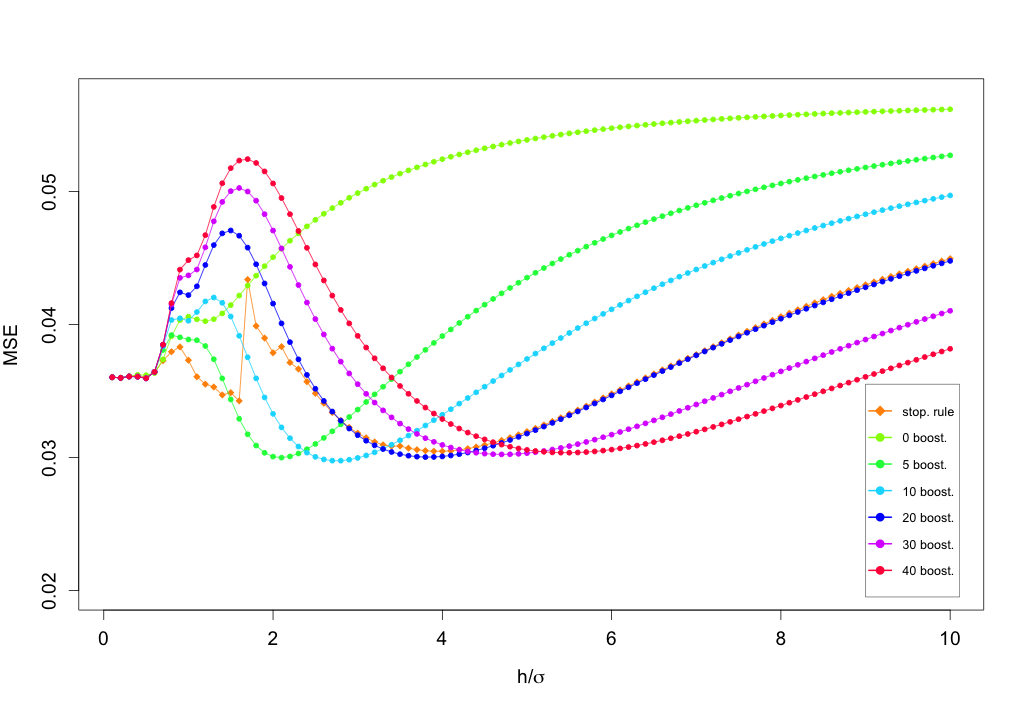}
        \caption{Model 5}
        \label{fig:model5_h}
    \end{subfigure}

     \begin{subfigure}[b]{0.7\textwidth}
        \includegraphics[width=\textwidth]{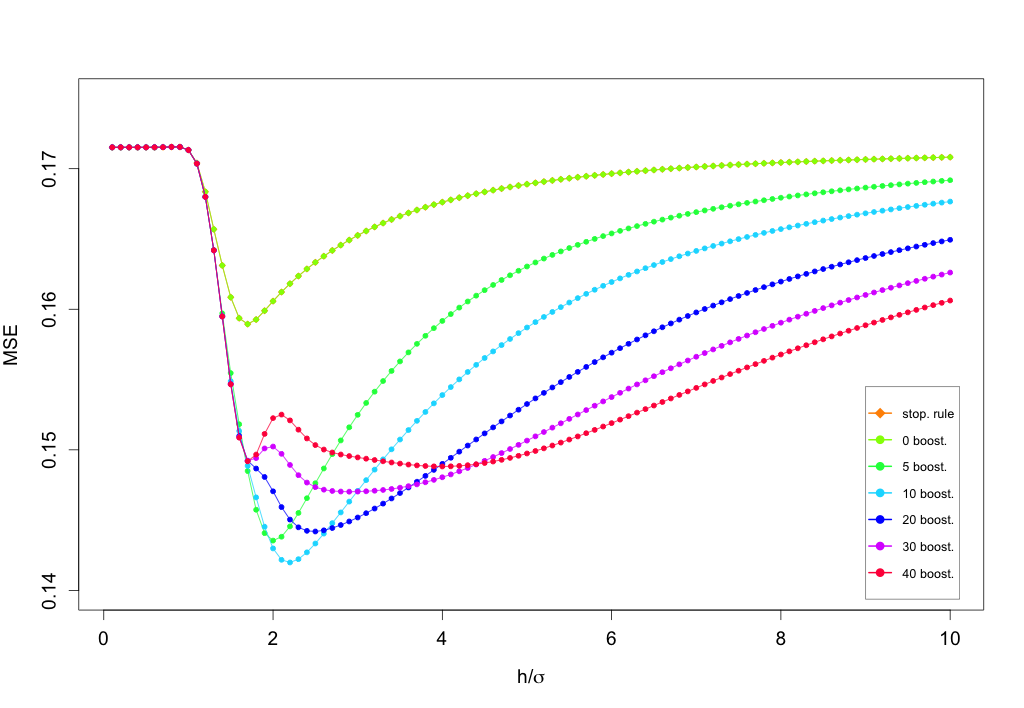}
        \caption{Model 6. The MSE of the estimate with stopping rules and the estimate without boosting coincide. In addition to
      $\hat{\sigma}$, the bandwidth of the multivariate model is scaled by
      a factor of 1.5. }
        \label{fig:model_multi_h}
    \end{subfigure}
\caption{Dependence on bandwidth $h$. MSE values, for models 5 and 6, for $h$
   between $0.1\hat{\sigma}$ 
   and $10\hat{\sigma}$ ($x$-axis with scaled $h/\hat{\sigma}$) for \mbox{MINT-T},
   without boosting, with fixed number of 
   boosting iterations, and with stopping rules. The time lag $p$ is set to 10.}\label{fig:bandwidth2} 
\end{figure}

We tested different bandwidths $h$ in the range of $0.1\hat{\sigma} - 10\hat{\sigma}$ for the simulated time series from Section~\ref{section:empirical}. In Figures~\ref{fig:bandwidth1} and \ref{fig:bandwidth2}, the MSE values are plotted against $h$.
We observe that with no
boosting, the performance is sensitive to the choice of the bandwidth. There is typically an optimal bandwidth if no boosting iterations are performed. For example, the optimal bandwidth is $h\approx0.6\hat{\sigma}$ for model 1 and $h\approx0.8\hat{\sigma}$ for model 5. The sensitivity largely disappears with increasing number of boosting iterations. Moreover, boosting is able to decrease the MSE. Therefore, we can simply take a
larger bandwidth in connection with subsequent boosting. We suggest the
following rule-of-thumb: $h=2\hat{\sigma}$ for univariate and $h=c_l\cdot
2\hat{\sigma}$ for multivariate time series along with $10$ boosting iterations. The factor $c_l$ is
approximately equal to $n^{\frac{1}{4+p}-\frac{1}{4+pl}}$ and corrects for
the dimensionality of the adjustment sets, i.e., $p$ for univariate and $pl$
for multivariate time series. For $n = 1000$, $l = 4$ and $p=10$, we used $c_l=1.5$. 
The observed results are consistent with what
we outlined in Section~\ref{section:implementation}: it is favourable to
choose large bandwidths, which results in an initial estimate with large
bias and small variance, and the bias is subsequently reduced by the
boosting iterations. 

In order to avoid unnecessary boosting iterations or "overboosting", we implemented the following stopping rules, which are also shown in Figures~\ref{fig:bandwidth1} and \ref{fig:bandwidth2}. First, we sum the absolute differences between two consecutive approximations:
\begin{equation} \label{eq:stoppingRule}
C(b):=\sum\limits_{i=1}^9 | (n-s-p)^{-1} \sum\limits_{k=s+p+1}^n \hat{g}_{R_b} (d_i, \textbf{X}_{k-s}^{\mathcal{S}}) | .
\end{equation}  
Then we terminate the boosting iterations if either the absolute difference $C(b)$ is smaller than 0.5\% of the previous estimate $\sum\limits_{i=1}^9 | (n-s-p)^{-1} \sum\limits_{k=s+p+1}^n \hat{m}_{b} (d_i, \textbf{X}_{k-s}^{\mathcal{S}}) |$ or if the proportion of two subsequent differences  $C(b)/C(b-1)$ is less than 75\% (see Section~\ref{section:implementation}). This way, we achieve little to no boosting for small bandwidths (small bias, high variance estimate) and more boosting iterations for larger bandwidths (high bias, small variance estimate). Apparently, our proposed stopping rule performs reasonably well in the univariate examples we
considered. For multivariate time series, the percentages used for the stopping rule should be adapted to the dimension of the time series. Additional simulations for the remaining models are presented
in Section~\ref{app:choiceh}.

\subsubsection{The choice of the time lag $p$ for adjustment}
\label{section:choicep}

\begin{figure}[!htb]
    \centering
    \begin{subfigure}[b]{0.45\textwidth}
        \includegraphics[width=\textwidth]{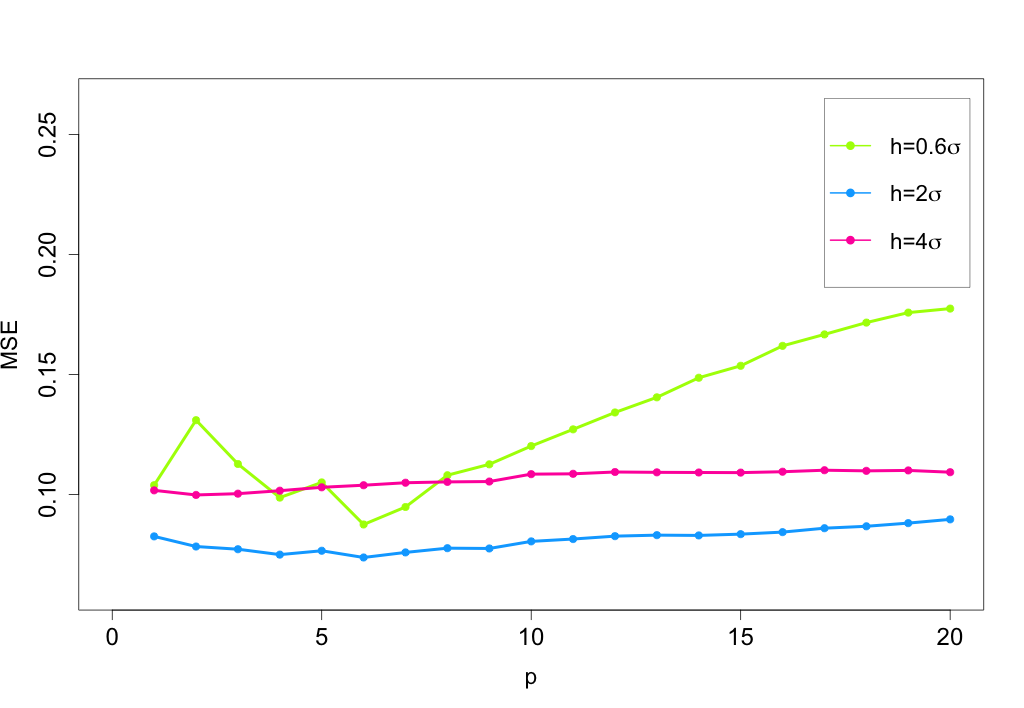}
        \caption{Model 1}
        \label{fig:model1_lag}
    \end{subfigure}
    \begin{subfigure}[b]{0.45\textwidth}
        \includegraphics[width=\textwidth]{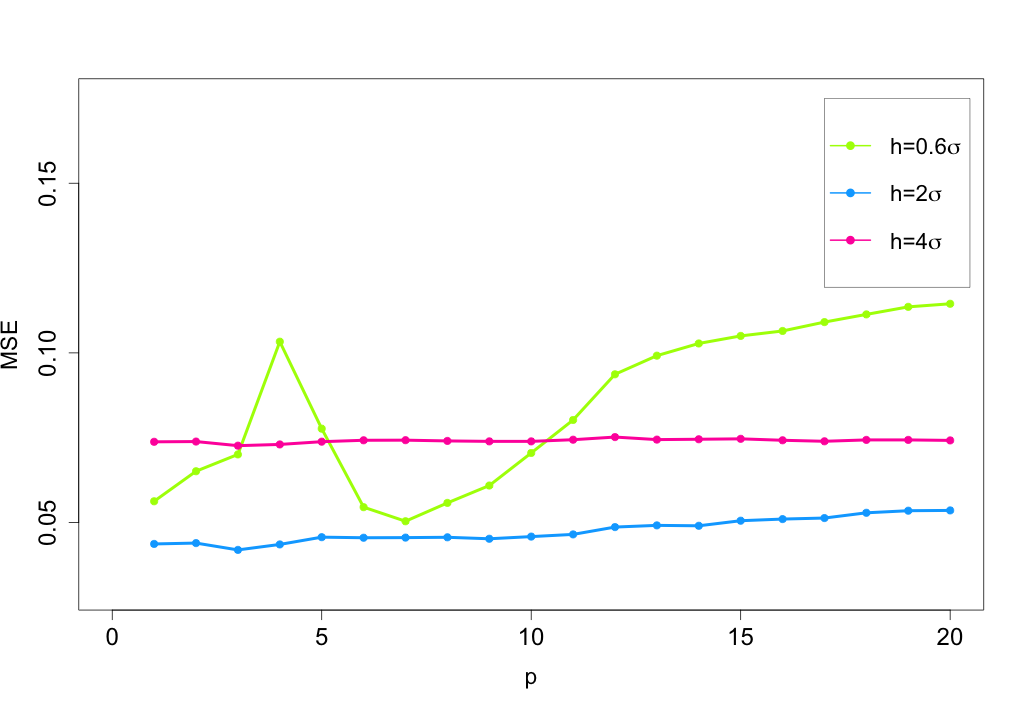}
        \caption{Model 2}
        \label{fig:model2_lag}
    \end{subfigure}
    \begin{subfigure}[b]{0.45\textwidth}
        \includegraphics[width=\textwidth]{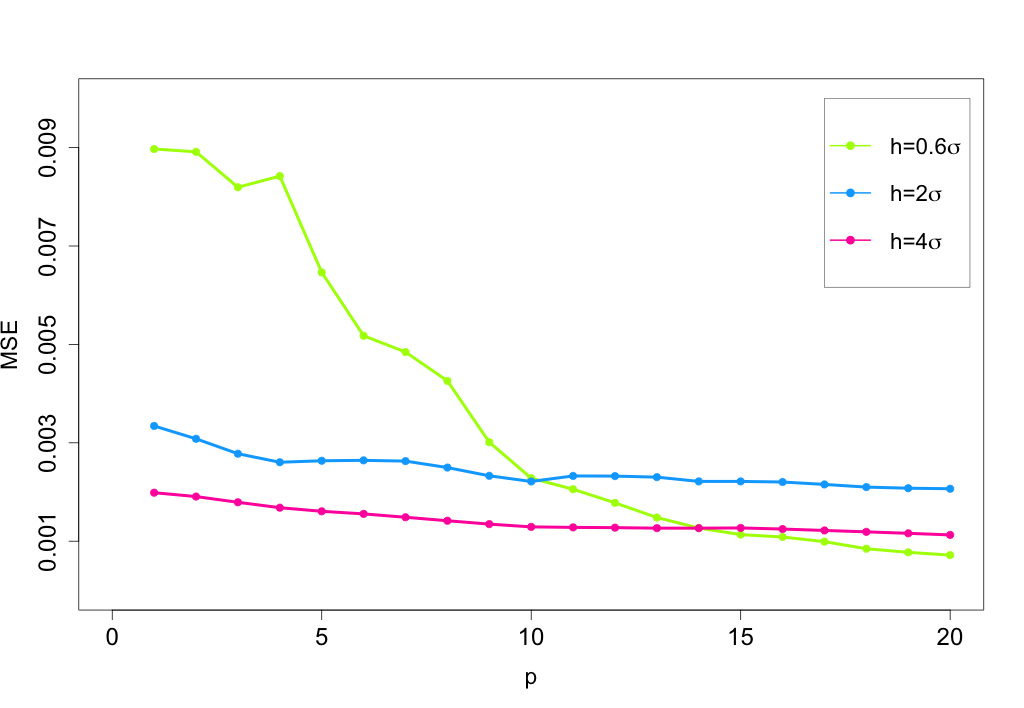}
        \caption{Model 3}
        \label{fig:model3_lag}
    \end{subfigure}
    \begin{subfigure}[b]{0.45\textwidth}
        \includegraphics[width=\textwidth]{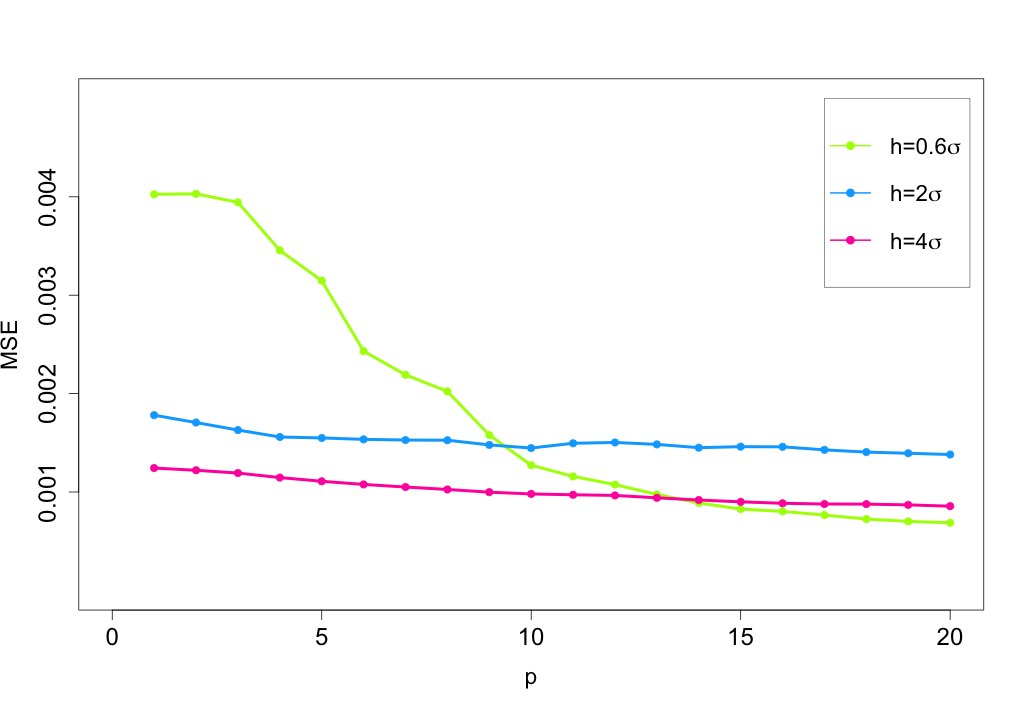}
        \caption{Model 4}
        \label{fig:model4_lag}
    \end{subfigure}
         \begin{subfigure}[b]{0.45\textwidth}
        \includegraphics[width=\textwidth]{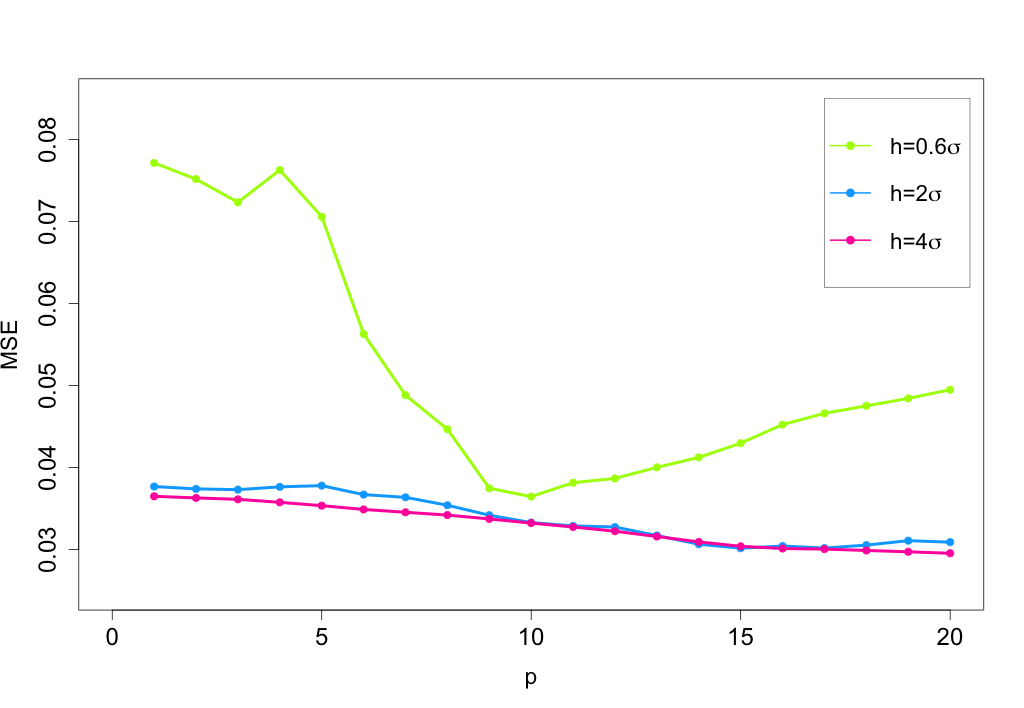}
        \caption{Model 5}
        \label{fig:model5_lag}
    \end{subfigure}
     \begin{subfigure}[b]{0.45\textwidth}
        \includegraphics[width=\textwidth]{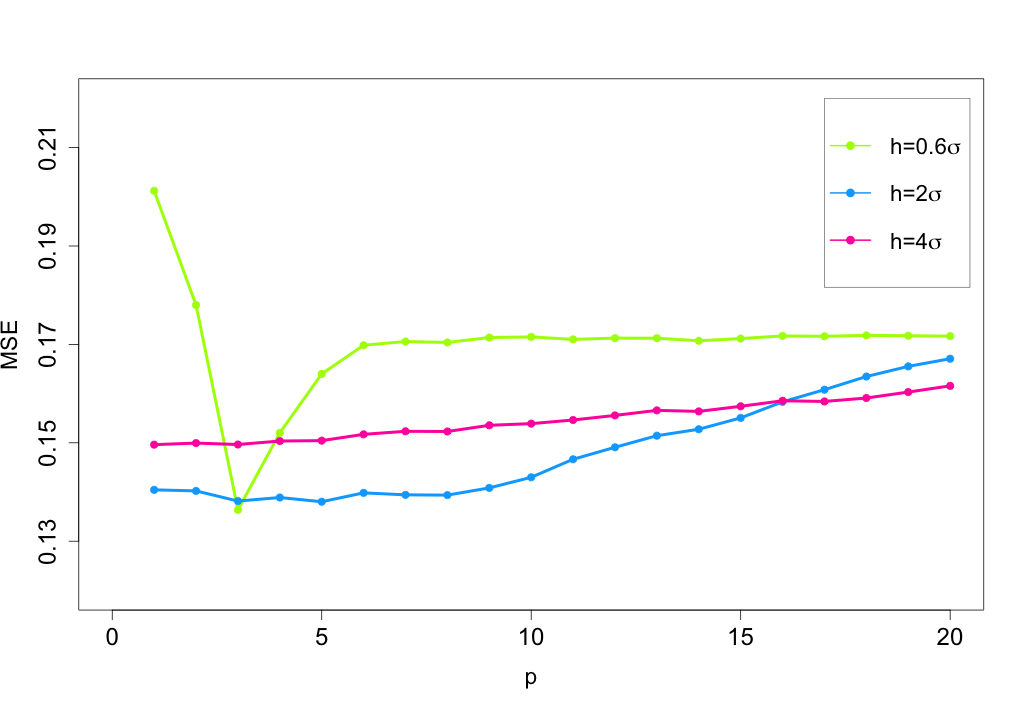}
        \caption{Model 6 }
        \label{fig:model_multi_lag}
    \end{subfigure}

\caption{Dependence on lag $p$. MSE values for $p$
   between $1$ 
   and $20$ for \mbox{MINT-T} with 10 boosting iterations. In addition to
      $\hat{\sigma}$, the bandwidth of the multivariate model is scaled by
      a factor of 1.5.}\label{fig:lag} 
\end{figure}

\mbox{MINT-T} requires an estimated time lag $p$ for the adjustment set indicated
with $\Superset$. We ran \mbox{MINT-T} on the simulated
time series from Section~\ref{section:empirical} for different values of
$p$ between 1 and 20. The corresponding MSE values for $h=0.6\hat{\sigma}$,
$h=2\hat{\sigma}$ and $h=4\hat{\sigma}$ are shown in Figure~\ref{fig:lag}. 

We observe in Figure~\ref{fig:lag} that for small bandwidths, e.g., $h=0.6\hat{\sigma}$, the performance is sensitive to the choice of $p$, while with larger bandwidths, e.g., $h=2\hat{\sigma}$, the sensitivity mostly disappears. The \mbox{reference} method deteriorates for
misspecified $p$, though. Therefore, particularly when \mbox{choosing} a large
bandwidth $h$, \mbox{MINT-T} is much more robust against model 
misspecification and rather insensitive to the choice of $p$. 

Our empirical results suggest that \mbox{MINT-T} is overall
surprisingly insensitive to the choice of the three tuning parameters, and
this constitutes a substantial practical advantage: we
should take a large bandwidth with sufficient amount of boosting iterations (we used $2 \hat{\sigma}$ for univariate
time series or $c_l \cdot 2 \hat{\sigma}$ for multivariate time series and $10$ boosting iterations), and then the choice of the lag
$p$ for adjustment does not matter much anymore (we suggest to inspect the
partial autocorrelation function of the time series or its transformed
value when considering the causal effect for a transformed response).

\subsection{Real data}

In this section we analyse financial data with \mbox{MINT-T}. Financial returns often show no evidence of serial correlation, however, when transformed, they often do. Therefore, it is more interesting to study the effects on a transformed response variable, e.g., $\E[X_{c_1,t}^2\,\vert\,\Do(X_{c_2,t-s}=d_i)]$. For each data set, monetary policy and currency data, we provide an estimated graph of the causal effects $\E[X_{c_1,t}^2\,\vert\,\Do(X_{c_2,t-s}=d_i)]$. These differ from the Markov graphs prevalent in the causal inference literature as the edges in the graphs represent substantial total causal effects instead of direct effects.

\subsubsection{Currency data}
\label{section:currency}
\begin{figure}[!htb]
    \centering
    \includegraphics[width=1\textwidth]{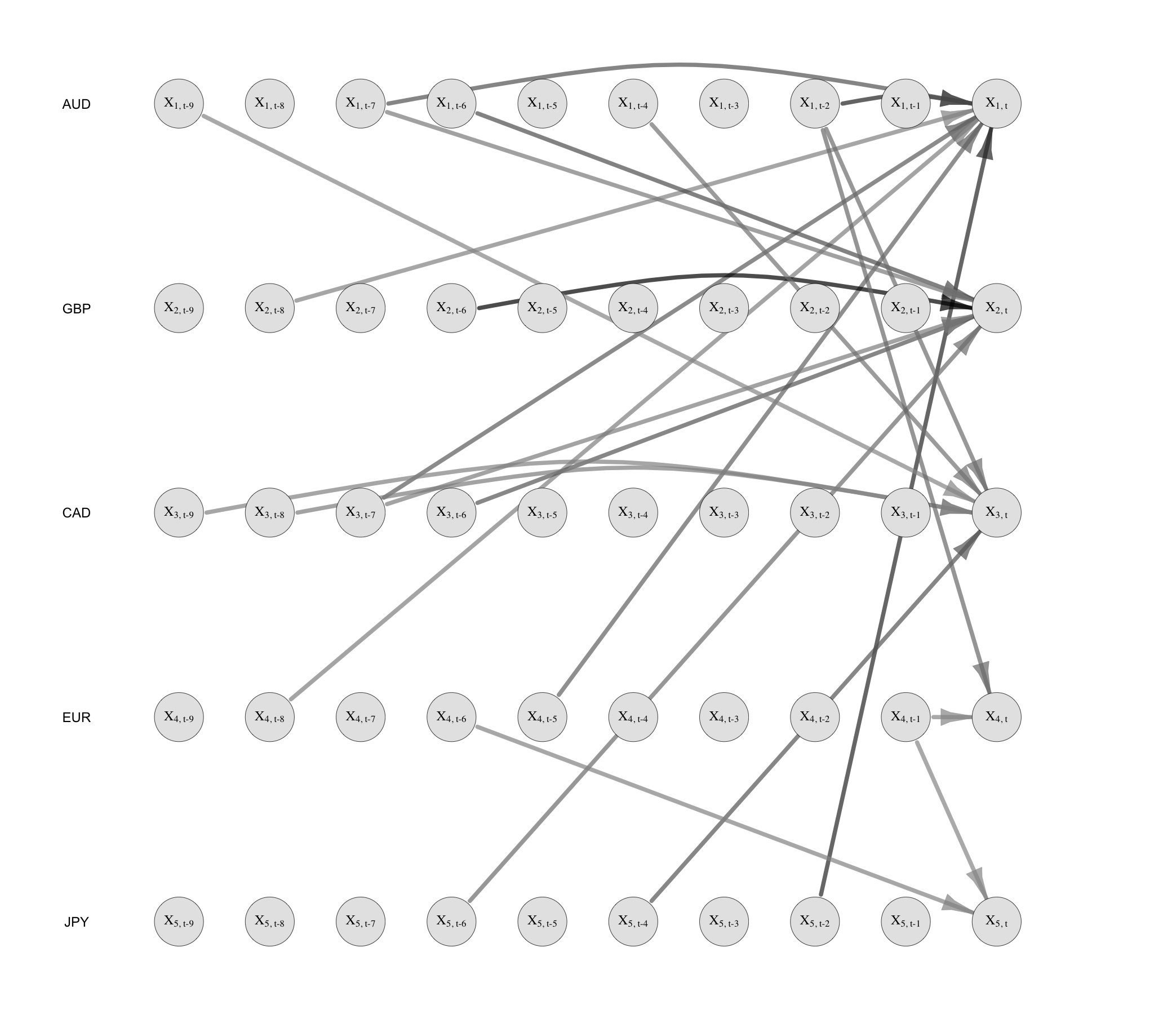}
    \caption{Currency data: log-returns of daily exchange rates of AUD,
      GBP, CAD, EUR and JPY vs. USD between January 4 1999 and October 15
      2010. An edge encodes a substantial estimated causal effect for the
      squared response $\E[X_{c_1,t}^2\mid\Do(X_{c_2,t-s}=x)]$, and its intensity is proportional to the strength of the estimated total causal
      effect.}
    \label{fig:meinshausen}
\end{figure}

We consider currency data containing the daily
exchange rates of five currencies versus US dollar from January 4 1999 to
October 15 2010. The time series components are AUD/USD, GBP/USD, CAD/USD,
EUR/USD and JPY/USD. We preprocessed the data by taking the log transform
and first order differencing. The resulting time series exhibits
heteroscedasticity and thus, it is worthwhile to study the effect of
interventions on the volatility function $\E[X_{c_1,t}^2\,\vert\,\Do(X_{c_2,t-s}=d_i)]$. We note that there is a correspondence with an intervention on the original currency value $P_{c_2,t-s}$. Since $X_{c_2,t-s} = \log(P_{c_2,t-s}) - \log(P_{c_2,t-s-1})$, we have that $\E[X_{c_1,t}^2\,\vert\,\Do(P_{c_2,t-s} = v)] =  \E[X_{c_1,t}^2\,\vert\,\Do(X_{c_2,t-s} = \log(v) - \log(P_{c_2,t-s-1}) )]$ with $v$ being the intervention value and $\log(P_{c_2,t-s-1})$ the observational log-price from one time-lag before the intervention takes place. 

We estimate the causal effects
$\E[X_{c_1,t}^2\,\vert\,\Do(X_{c_2,t-s}=d_i)]$ for every $s = 1,\ldots,9$,
$c_1,c_2=1,\ldots,l $ and the nine deciles of the time series
$d_1,\ldots,d_9$. The tuning parameters are chosen as $p=10$, $B=10$ and $h=3\hat{\sigma}$. 
\begin{sloppypar}We represent these causal effects in a graph with nodes
corresponding to the random variables from the $l$ components and the time indices $t, t-1, \ldots, t-9$. We draw an edge from the node
corresponding to $X_{c_2,t-s}$ to $X_{c_1,t}$ if the relative strength of the causal effects $CS_{c_1,c_2}(s):=(\sum_{i=1}^9\vert\hat{\E}[X_{c_1,t}^2\,\vert\,\Do(X_{c_2,t-s}=d_i)]-E[X_{c_1,t}^2]\vert - \frac{1}{9}\sum_{s=1}^9 \sum_{i=1}^9\vert\hat{\E}[X_{c_1,t}^2\,\vert\,\Do(X_{c_2,t-s}=d_i)]-E[X_{c_1,t}^2]\vert)/E[X_{c_1,t}^2]$ exceeds a
threshold. We subtracted $\frac{1}{9}\sum_{s=1}^9
\sum_{i=1}^9\vert\hat{\E}[X_{c_1,t}^2\,\vert\,\Do(X_{c_2,t-s}=d_i)]-E[X_{c_1,t}^2]\vert
$ from the causal strength to balance the values across the different time series components. We set the threshold
to the ninth decile of the values in the set $\left\{
  CS_{c_1,c_2}(s)\,\vert\, 
 c_1,c_2=1,\ldots,5,\,s = 1,\ldots,9 \right\}$. The resulting graph is shown in
Figure~\ref{fig:meinshausen}. The intensity of an edge is
proportional to the magnitude of the
  values in $ CS_{c_1,c_2}(s)$.   
  \end{sloppypar}

In Figure~\ref{fig:meinshausen}, the exchange rates of AUD, GBP, CAD and EUR are each affected by their previous values. If we intervene on a currency exchange rate, we expect a change in the demand for the currency, which affects the exchange rate at the following time points. Furthermore, we observe edges linking different components in Figure~\ref{fig:meinshausen}. This implies that an intervention at one currency not only alters its own future value, but the future values of other currencies as well. For example, the exchange rates between AUD, GBP and CAD are closely connected, which can be explained by economical bonds and political links between the commonwealth countries.

We described in Section~\ref{sec:instant} that MINT-T is able to
  retrieve total causal links between lagged variables in the presence of
  instantaneous effects as long as the time lag from the intervention to
  the target is at least one. The time
  resolution is one day in the currency data set. Therefore, instantaneous
  effects cannot be ruled out completely. We reestimated the causal graph,
  this time taking instantaneous effects into account. The graph in
  Figure~\ref{fig:meinshausen_instant} was drawn in the same way described
  before with exactly the same choice of
  tuning parameters but estimating the 
  quantity $\tilde{\E}[X_{c_1,t}^2\,\vert\,\Do(X_{c_2,t-s}=d_i)]$ instead
  of $\E[X_{c_1,t}^2\,\vert\,\Do(X_{c_2,t-s}=d_i)]$. The differences
  between Figure~\ref{fig:meinshausen} and
  Figure~\ref{fig:meinshausen_instant} are subtle. In
  Figure~\ref{fig:meinshausen_instant} the edges within the same currency
  are more pronounced and JPY is more susceptible to interventions to GBP
  and EUR than in Figure~\ref{fig:meinshausen}. 

\begin{figure}[!htb]
    \centering
    \includegraphics[width=1\textwidth]{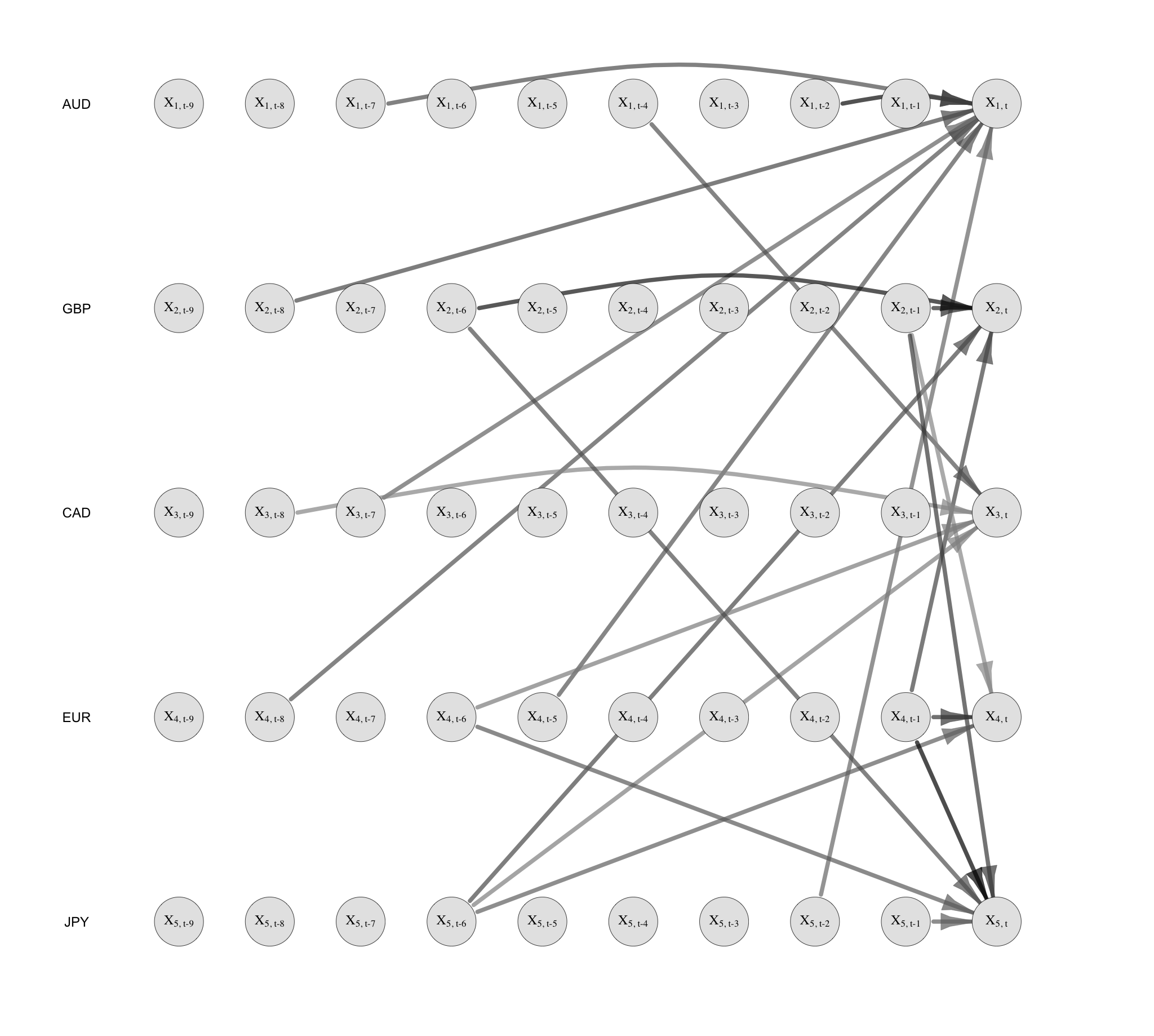}
    \caption{Currency data: log-returns of daily exchange rates of AUD,
      GBP, CAD, EUR and JPY vs. USD between January 4 1999 and October 15
      2010. An edge encodes a substantial estimated effect for the
      squared response $\tilde{ \E}[X_{c_1,t}^2\mid\Do(X_{c_2,t-s}=x)]$
      with instantaneous effects.}
    \label{fig:meinshausen_instant}
\end{figure}
\bigskip

\subsubsection{Macroeconomic data}
\label{section:macro}

\begin{figure}[!htb]
    \centering
    \includegraphics[width=1\textwidth]{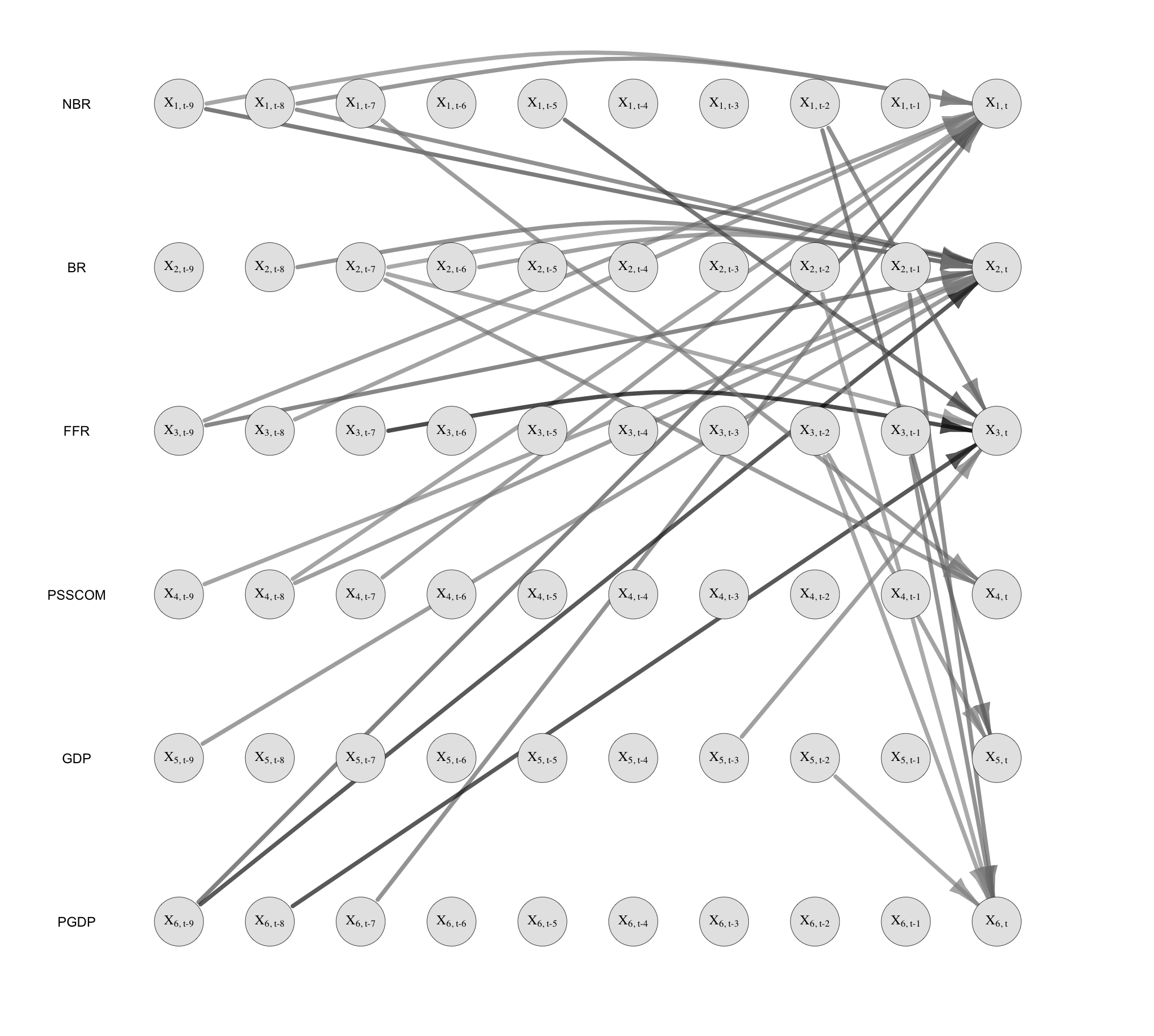}
    \caption{Macroeconomic data \cite{bernanke1998}, consisting of monthly
      observations of $NBR_t$, $BR_t$, $FFR_t$, $PSCCOM_t$, $GDP_t$,
      $PGDP_t$ between January 1965 and December 1996 (first differences of
      log-transformed values, see text). An edge
      encodes a substantial estimated causal effect for the 
      squared response $\E[X^2_{c_1,t}\mid\Do(X_{c_2,t-s}=x)]$, and its
      intensity is proportional to the strength of the estimated total 
      causal effect.} 
    \label{fig:moneta}
\end{figure}
Next, we look at macroeconomic data provided by \citet{bernanke1998}. The data set contains six monthly US time series from January 1965 to December
1996. The components are non-borrowed reserves and extended credit $NBR_t$,
borrowed reserves $BR_t$, the federal funds rate $FFR_t$, the Dow-Jones
index of spot commodity prices $PSCCOM_t$, real gross domestic product
$GDP_t$ and the gross domestic product deflator $PGDP_t$. The variables can
be grouped into policy ($BR_t$, $NBR_t$, $FFR_t$) and macroeconomic variables
($GDP_t$, $PGDP_t$, $PSCCOM_t$). The data was preprocessed by taking the
log transform and differencing. Due to heteroscedasticity, we focus on the
effect of interventions on the volatility function. 

We estimate the graph as described in Section~\ref{section:currency} with the same choice of tuning parameters, and
the result is shown in Figure~\ref{fig:moneta}. The variable that
is influenced most by interventions is $FFR_t$. $FFR_t$ is often described
as a key indicator of monetary policy shocks
\cite{moneta2013causal}.
In the latter reference a parametric model is considered allowing also for
instantaneous effects. We will permit instantaneous effects in a
nonparametric setting when using the procedure from Section
\ref{sec:instant}, and the results are given in Figure
\ref{fig:moneta_instant}.
By law, banks are required to hold a minimum fraction of customer deposits as reserves at
the federal reserve. If banks own less than the minimum fraction, they may
choose to borrow the required amount from another bank or the federal
reserve. Otherwise, they may lend excessive reserves to other banks. The
federal funds rate is the interest rate at which banks trade balances held
at the federal reserve. The causal links between $NBR_t$, $BR_t$ and $FFR_t$ reflect this
relationship. For example, $FFR_{t}$ responds to the changes in demand for borrowed and non-borrowed reserves. Furthermore, the federal reserve observes macroeconomic variables in order to regulate the federal funds rate through open market operations. Hence,
$FFR_t$ is also targeted by interventions on non-policy variables such as income ($GDP_t$) and price level ($PGDP_t$). 

Since the time resolution of the macroeconomic data set is one month,
  we reestimated the graph for the macroeconomic data taking potential
  instantaneous effects into account. The graph in
  Figure~\ref{fig:moneta_instant} is 
  based on estimates of $\tilde{\E}[X_{c_1,t}^2\,\vert\,\Do(X_{c_2,t-s}=d_i)]$
  instead of $\E[X_{c_1,t}^2\,\vert\,\Do(X_{c_2,t-s}=d_i)]$ while the
  tuning parameters are unchanged. Similar to the currency data, the
  differences between Figure~\ref{fig:moneta} and
  Figure~\ref{fig:moneta_instant} are subtle. For example, the strongest
  causal links from Figure~\ref{fig:moneta}, i.e., $NBR_{t-5} \rightarrow
  FFR_{t}$,  $FFR_{t-7} \rightarrow FFR_{t}$, $PGDP_{t-8} \rightarrow
  FFR_{t}$ and $PGDP_{t-9} \rightarrow BR_{t}$ remain the strongest links
  in Figure~\ref{fig:moneta_instant}. 
\begin{figure}[!htb]
    \centering
    \includegraphics[width=1\textwidth]{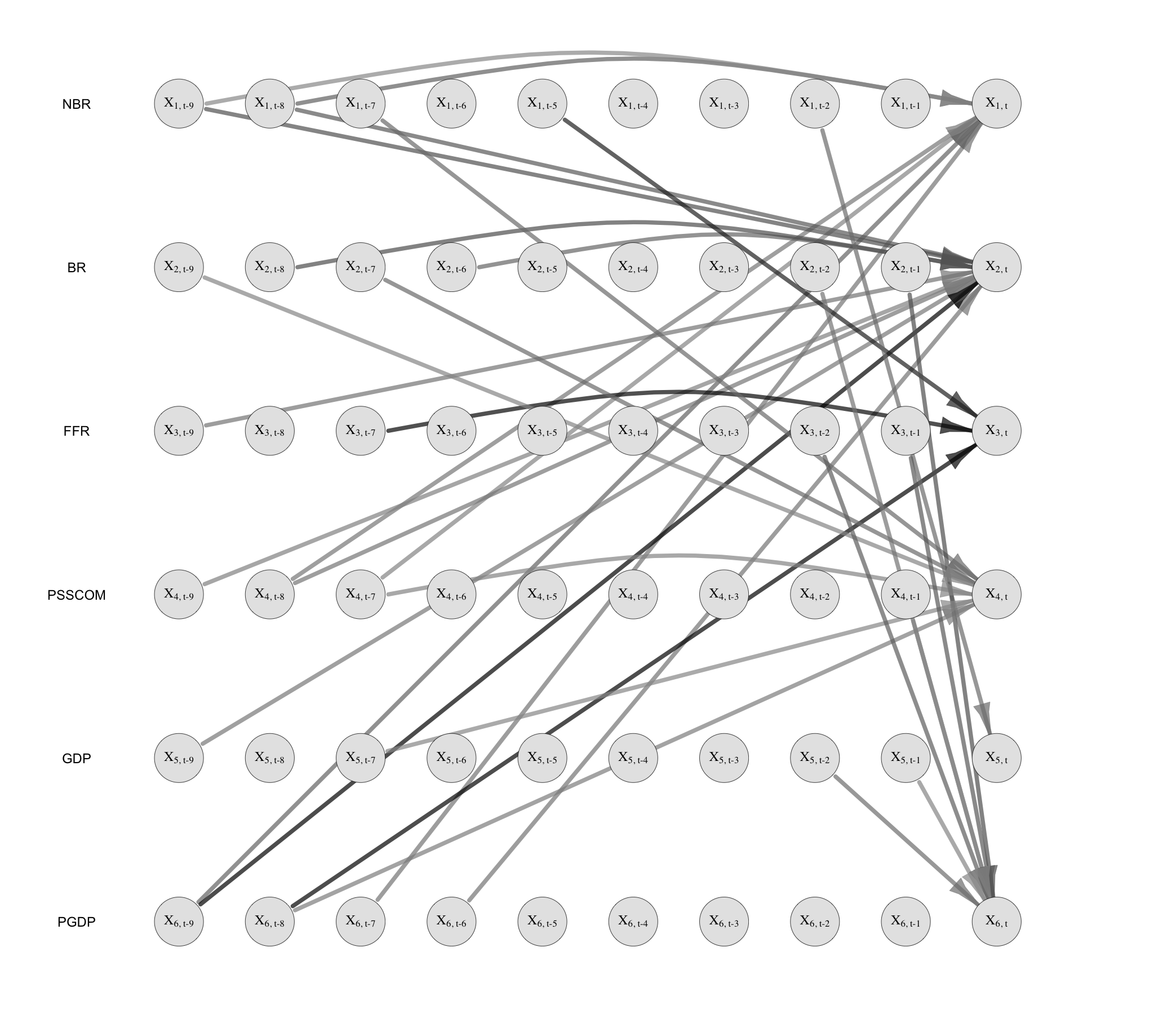}
    \caption{Macroeconomic data \cite{bernanke1998}, consisting of monthly
      observations of $NBR_t$, $BR_t$, $FFR_t$, $PSCCOM_t$, $GDP_t$,
      $PGDP_t$ between January 1965 and December 1996 (first differences of
      log-transformed values, see text). An edge encodes a substantial estimated effect for the
      squared response $\tilde{ \E}[X_{c_1,t}^2\mid\Do(X_{c_2,t-s}=x)]$
      with instantaneous effects.} 
    \label{fig:moneta_instant}
\end{figure}

\section{Conclusions}

Within the framework of stationary Markovian processes, we considered a
simple method based on observational time series data to 
infer the effect of interventions. We showed 
that a marginal integration estimator, called \mbox{MINT-T}, recovers the true
intervention effect with optimal nonparametric rate $n^{-2/5}$ under some 
regularity conditions and assuming no instantaneous effects in
  multivariate settings. This
is the optimal convergence rate for the 
estimation of a one-dimensional twice-differentiable function. Even
though it is infeasible to estimate the data generating stochastic process
in a nonparametric way, \mbox{MINT-T} is fully nonparametric and remains largely unaffected
by the curse of dimensionality assuming smoothness and additional
  regularity conditions. The advantage of our method is that we do
  not require knowledge of an adjustment set for causal effects: instead,
  we only need to specify an upper bound for the order of the underlying
  Markovian process. Even in presence of time-instantaneous effects, 
the methodology is shown to provide interesting results, avoiding false
positive statements. 
  
  Double robust methods require the correct specification of either the
  regression model or the propensity score model for consistent estimation
  of $\E[Y\,\vert\,\Do(X=x)]$ (cf.~\citet{van2003}). Typically, $X$ is a
  binary treatment variable. \mbox{MINT-T}, on the other hand, considers
  continuous intervention variables and is fully nonparametric as it does
  not require the specification of any model but assumes a
    sufficient amount of smoothness. We refer to
  \citet{ernest2015} for a more in-depth comparison of marginal integration
  and doubly robust methods.

Our theoretical assumptions include smoothness and the use of
  higher order kernel: we found that MINT-T, with an implementation
    based on boosting instead of using an explicit higher-order kernel,
    performed well in simulations with 
    smooth underlying conditional mean functions. We compared
\mbox{MINT-T} to a reference method: it fits a generalised 
additive model and infers the causal effect via simulation, as if the data
were generated from such an additive model. If the model is misspecified,
the reference method is inconsistent: this is in contrast to \mbox{MINT-T}
which does not depend on the specification of a time series model. In our
empirical studies, \mbox{MINT-T} outperformed the reference on all tested models except for some univariate $\AR$-models.
In addition, \mbox{MINT-T} provided an acceleration of the computational
time by a factor of 85 on average and in fact, \mbox{MINT-T} is computationally
efficient and feasible for multivariate time series in potentially
large-dimensional problems. 

Inferring the causal effect from readily available observational time series
data can offer helpful guidelines for researchers who wish to design
experiments before committing to irreversible and comprehensive
interventions. \mbox{MINT-T} provides a feasible, fully nonparametric tool
for this task.

\appendix
\section{Proof of Theorem~\ref{thm1}}\label{proofthm1}
\begin{proof}[Sketch of the proof]
	\begin{sloppypar}The proof follows immediately from Theorem 1 and Remark 3 \cite{fan1998} by including some modifications for dependent variables and choosing the weight function $W(\cdot)$ to be identical to one. In Equation (6.1) in \cite{fan1998}, we apply Theorem 2.21 in \cite{fan2005}. In Equation (6.4) and on p.~962 in~\cite{fan1998}, we replace the uniform convergence of kernel density estimators for the i.i.d.\ case by a similar result for dependent variables in \cite{hansen2008}. In order to obtain the final result $\hat{\E}(X_{c_1,t}\,\vert\, \Do(X_{c_2,t-s}=x))-\E(X_{c_1,t}\,\vert\, \Do(X_{c_2,t-s}=x))=O(h_1^2)+O_p(1/\sqrt{nh_1})$, it remains to show that each of the following terms
		$$n^{-1}\sum_{j=1}^n \epsilon_j^* K_{h_1}(X_{c_2,j-s}-x),$$
		$$T_{n,1}=n^{-2}\sum_{i\neq j}\Gamma(\mathbf{X}_{i-s}^{\Superset})\tilde{r}_{ij}, $$
		$$T_{n,2}=n^{-2}\sum_{i\neq j} W(\mathbf{X}_{i-s}^{\Superset}) \tilde{r}_{ij} p^{-1}(\mathbf{x}^{i-s})p^{(1,0)}(\mathbf{x}^{i-s})^T \mu_2(K)(X_{c_2,j-s}-x), $$
		$$T_{n,3}=n^{-2}\sum_{i\neq j}\epsilon_jK_{h_1}(X_{c_2,j-s}-x)V_{ij} \text{ and}$$
		$$T_{n,4}=n^{-2}\sum_{i\neq j} W(\mathbf{X}_{i-s}^{\Superset}) A_j(\mathbf{x}^{i-s})\epsilon_{j}p^{-1}(\mathbf{x}^{i-s})p^{(1,0)}(\mathbf{x}^{i-s})^T \mu_2(K)(X_{c_2,j-s}-x) $$
		is of order $O(h_1^2)+O_p(1/\sqrt{nh_1})$. Here,
                $\mathbf{x}^{i-s}=(x, \mathbf{X}_{i-s}^{\Superset})$,
                $\epsilon_j^*=\Gamma(\mathbf{X}_{j-s}^{\Superset})p(\mathbf{X}_{j-s}^{\Superset})(X_{c_1,j}-m(X_{c_2,j-s},\mathbf{X}_{j-s}^{\Superset}))$,
               where $m(\cdot)$ denotes the regression function and
                $\Gamma(\mathbf{X}_{j-s}^{\Superset})=W(\mathbf{X}_{j-s}^{\Superset})/p(\mathbf{x}^{i-s})$. Furthermore,
                $\hat{r}_{ij}=m(X_{c_2,j-s},\mathbf{X}_{j-s}^{\Superset})-m(x,\mathbf{X}_{i-s}^{\Superset})-f'_1(x)^T(X_{c_2,j-s}-x)$,
                $\tilde{r}_{ij}=A_j(\mathbf{x}^{i-s})\hat{r}_{ij}-\E
                A_j(\mathbf{x}^{i-s})\hat{r}_{ij}$ and
                $A_j(\mathbf{x}^{i-s})=K_{h_1}(X_{c_2,j-s}-x)L_{h_2}(\mathbf{X}_{j-s}^{\Superset}-\mathbf{X}_{i-s}^{\Superset})$. This
                can be achieved by calculating the first and second moments
                and invoking the covariance bound in Proposition 2.5(ii)
                in \citet{fan2005}. A detailed proof is given next.\end{sloppypar}
\end{proof}

\begin{proof}[Proof of Theorem~\ref{thm1}]
	The proof follows along the lines of the proof of Theorem 1 in \citet{fan1998} by replacing $Y$ with $X_{c_1,t}$, $X_1$ with
        $X_{c_2,t-s}$, $X_2$ by $\mathbf{X}_{t-s}^{\mathcal{S}}$ and $x_1$
        with $x$. For simplicity of notation, we shall neglect the discrete
        variable $X_3$. To avoid confusion, we will keep the notation from
        \citet{fan1998}. For example, $p$ will refer to the dimension of
        the variable $X_1$ instead of the time lag within this proof. Also,
        $\E_j(\cdot)$ will denote the expectation with respect to all
        variables except for $X_j$. In our case, the weight function
        $W(\cdot)$ \cite{fan1998} is identical to 1. Assumption~A.1
        \cite{fan1998} is thereby satisfied. In order to adapt Theorem 1
        \cite{fan1998} for dependent variables, we require the following
        modifications. 

	In Equation (6.1), we apply Theorem 2.21 from \citet{fan2005}. The
        variables $m(x^i)W(X_{2i})$ are bounded due to Assumption~1.6 and
        Assumption~A.1 \cite{fan1998}. Also,
        $\sum_{k\geq1}\alpha_k<\infty$ due to Assumption~1.1. Hence,
        condition (ii) of Theorem 2.21 \cite{fan2005} is satisfied. 

	In Equation (6.4) and on p. 962 of \citet{fan1998}, we replace the
        uniform convergence of kernel density estimators for the
        i.i.d. case by Theorem 2 in \citet{hansen2008} for dependent
        variables. 

	Thus, we can rewrite Equation (6.5) \cite[p.963]{fan1998} as
	\begin{equation}
	\hat{f}_1^*(x_1)-f_1^*(x_1)=O(h_1^2)+o_p(1/\sqrt{nh_1^p})+T_{n,1}+T_{n,2}+T_{n,3}+T_{n,4}+n^{-1}\sum_{j=1}^n \epsilon_j^* K_{h_1}(X_{1j}-x_1), \nonumber
	\end{equation}
	where $\epsilon_j^*=\Gamma(X_{2j})p(X_{2j})\epsilon_j$, $\Gamma(X_{2j})=W(X_{2j})/p(x^j)$ and  $\epsilon_j=Y-m(X_j)$, where $m(\cdot)$ denotes the regression function. $\epsilon_j^*$, $\Gamma(\cdot)$ are bounded due to Assumptions 1.2-1.3, 1.6 and Assumption~A.1 \cite{fan1998}.  We will show that each of the remaining terms
	$$n^{-1}\sum_{j=1}^n \epsilon_j^* K_{h_1}(X_{1j}-x_1),$$
	$$T_{n,1}=n^{-2}\sum_{i\neq j}\Gamma(X_{2i})\tilde{r}_{ij}, $$
	$$T_{n,2}=n^{-2}\sum_{i\neq j} W(X_{2i}) \tilde{r}_{ij} p^{-1}(x^i)p^{(1,0)}(x^i)^T \mu_2(K)(X_{1j}-x_1), $$
	$$T_{n,3}=n^{-2}\sum_{i\neq j}\epsilon_jK_{h_1}(X_{1j}-x_1)V_{ij} \text{ and}$$
	$$T_{n,4}=n^{-2}\sum_{i\neq j} W(X_{2i}) A_j(x^i)\epsilon_{j}p^{-1}(x^i)p^{(1,0)}(x^i)^T \mu_2(K)(X_{1j}-x_1) $$
	is of order $O_p(1/\sqrt{nh_1^p})$. Here, $V_{ij}=\Gamma(X_{2i})L_{h_2}(X_{2j}-X_{2i})-p(X_{2j})\Gamma(X_{2j})$, $\mu_2(K)=\int u^2 K(u) du$, $\tilde{r}_{ij}=A_j(x^i) \hat{r}_{ij} - \E A_j(x^i) \hat{r}_{ij}$ and $\hat{r}_{ij}=m(X_j)-m(x^i)-f'_1(x_1)^T (X_{1j}-x_1)$.
	With this goal in mind, we calculate the first and second moments of the sums. The first moments
	$\E \epsilon_j^* K_{h_1}(X_{1j}-x_1)=0$, $\E T_{n,3}=0$ and  $\E T_{n,4}=0$ as $\E(\epsilon_{j}\,\vert\,X_j)=0$, $\E_j (V_{ij}\,\vert\,X_j)=O(1)$ and  $\E_j(L_{h_2}(X_{2j}-X_{2i})\,\vert\, X_j)=O(1)$ almost surely (see Lemma~\ref{lemma1}) and $W(X_{2i}),\,p^{-1}(x^i),\,p^{(1,0)}(x^i),\,X_{1j}$ are bounded functions due to Assumptions 1.2-1.3 and Assumption~A.1 \cite{fan1998}. Also, $\E T_{n,1}=0$ and $\E T_{n,2}=0$ because $\E \tilde{r}_{ij}=0$ and $W(X_{2i}),\,p^{-1}(x^i),\,p^{(1,0)}(x^i),\,X_{1j}$ are bounded functions.
	By Theorem 14.4-1 \cite{bishop1975}, it is now sufficient to show that the second moments are of order $O(n^{-1}h_1^{-p})$. The calculation of the second moment is slightly more involved, and throughout our derivations $C$ will refer to a positive but not necessarily the same constant.
	\begin{eqnarray}
	&&\E \left(n^{-1}\sum_{j=1}^n \epsilon_j^* K_{h_1}(X_{1j}-x_1)\right)^2 \nonumber\\
	&\leq&Cn^{-2}\sum_{i,j}\E \left( \epsilon_j K_{h_1}(X_{1j}-x_1) \epsilon_i 	K_{h_1}(X_{1i}-x_1) \right)\nonumber\\
	&\leq&C n^{-1}\sum_{k=1}^{n-1} \left(1-\frac{k}{n}\right) \Cov \left( \epsilon_1 K_{h_1}(X_{11}-x_1), \epsilon_{k+1} K_{h_1}(X_{1k+1}-x_1) \right)+C n^{-1}\Var(\epsilon_1 K_{h_1}(X_{11}-x_1)) \nonumber\\
	&\leq&C n^{-1}\sum_{k=1}^{n-1}  \Cov \left( \epsilon_1 K_{h_1}(X_{11}-x_1), \epsilon_{k+1} K_{h_1}(X_{1k+1}-x_1) \right) +O(n^{-1}h_1^{-p}) \nonumber
	\end{eqnarray}
	In the last line, we used
	\begin{eqnarray}
	\Var(\epsilon_1 K_{h_1}(X_{11}-x_1)) &=&\E(\epsilon_1^2 K^2_{h_1}(X_{11}-x_1))\nonumber\\
	&\leq&C\E( K^2_{h_1}(X_{11}-x_1))\nonumber\\
	&=&C\int K^2_{h_1}(X_{11}-x_1)p(X_{11})dX_{11}\nonumber\\
	&=&Ch_1^{-p}\int K^2(u)p(x_1+h_1 u)du\nonumber\\
	&=&Ch_1^{-p}p(x_1)\int K^2(u)du+o(h_1^{-p})\nonumber\\
	&=&O(h_1^{-p}).\label{eq:kernelvar}
	\end{eqnarray}
	The covariance term can be bounded by a constant. 
	\begin{eqnarray}
&&\Cov \left( \epsilon_1 K_{h_1}(X_{11}-x_1), \epsilon_{k+1} K_{h_1}(X_{1k+1}-x_1) \right)\nonumber\\
&\leq& C \E \left( K_{h_1}(X_{11}-x_1) K_{h_1}(X_{1k+1}-x_1) \right)\nonumber\\
&\leq& C \sup p(X_{11},X_{1k+1}) \leq C_1 \nonumber
	\end{eqnarray}
In the last step, we used Assumption~1.5. On the other hand, the covariance between two bounded, random variables $X_{1j}$ and $X_{1l}$ is bounded by a constant times the mixing coefficient $\alpha_{\vert j-l\vert}$ according to Proposition 2.5(ii) in \citet{fan2005}. The covariance bound together with Assumption~1.7 yields
$$\Cov \left( \epsilon_1 K_{h_1}(X_{11}-x_1), \epsilon_k K_{h_1}(X_{1k+1}-x_1) \right)\nonumber \leq C_2 h_1^{-2p} \alpha_k.$$
Let $b_n \rightarrow \infty$ be a sequence of integers. Then,
\begin{eqnarray}
&&\sum_{k=1}^{n-1}  \Cov \left( \epsilon_1 K_{h_1}(X_{11}-x_1), \epsilon_{k+1} K_{h_1}(X_{1k+1}-x_1) \right)\label{eq:trick1}\\
&=&\sum_{k=1}^{b_n-1}  \Cov \left( \epsilon_1 K_{h_1}(X_{11}-x_1), \epsilon_{k+1} K_{h_1}(X_{1k+1}-x_1) \right)\nonumber\\
&+&\sum_{k=b_n}^{n-1}  \Cov \left( \epsilon_1 K_{h_1}(X_{11}-x_1), \epsilon_{k+1} K_{h_1}(X_{1k+1}-x_1) \right)\nonumber\\
&\leq&\sum_{k=1}^{b_n-1}  C_1 +\sum_{k=b_n}^{n-1}   C_2 h_1^{-2p} \alpha_k \nonumber\\
&\leq&O(b_n)  +\sum_{k=b_n}^{n-1}   C_2 h_1^{-2p} k^{-\beta} \nonumber\\
&=&O(b_n)  +O(b_n^{-\beta+1} h_1^{-2p}) \label{eq:trick2}
\end{eqnarray}
Choosing $b_n=h_1^{-2p/\beta}$ gives us the desired rate
$O(h_1^{-2p/\beta})=o(h_1^{-p})$ for $\beta>2$.

We now turn to $T_{n,1}$ and $T_{n,2}$. Both $T_{n,1}$ and $T_{n,2}$ are smaller than a constant times $n^{-2}\sum_{i\neq j}\tilde{r}_{ij}$. Therefore,
\begin{equation}
\E T_{n,1}^2 \text{ and } \E T_{n,2}^2\leq C n^{-4}\sum_{i\neq j}\sum_{k\neq l}\E \left(\tilde{r}_{ij}\tilde{r}_{kl}\right). \nonumber 
\end{equation}
For four different indices, i.e., $k,l \not\in \{i,j\} $, the summand is zero because $\E\tilde{r}_{ij}=0$ and $\tilde{r}_{ij}=O(h_1^2+h_2)$ (see \cite[p.965]{fan1998} and Lemma~\ref{lemma2}). For three different indices, the sum is at most of order $O(n^{-1}(h_1^4+h_2^2))$ and for two different indices of order $O(n^{-2}(h_1^4+h_2^2))$.
Hence, we obtain the desired rate $o(n^{-1})$.

It remains to show that $\E T_{n,3}^2$ and $\E T_{n,4}^2$ are $O(n^{-1}h_1^{-p})$. By performing a variable transformation and a Taylor expansion one can see that $\E_j(V_{ij}\,\vert\, X_j, X_k, X_l)=O(1)$, $\E_j(V_{ij}^2\,\vert\, X_j)=O(h_2^{-q})$ almost surely (see Lemma~\ref{lemma1}). In addition, $V_{ij}=O(h_2^{-q})$ due to Assumption~1.2-1.3, 1.7 and Assumption~A.1 \cite{fan1998}. We will treat the cases of two, three and four different indices separately. First for four different indices $k,l \not\in \{i,j\} $
\begin{eqnarray}
\E T_{n,3}^2&=&n^{-4}\sum_{i\neq j}\sum_{k\neq l} \E \left( \epsilon_jK_{h_1}(X_{1j}-x_1)V_{ij} \epsilon_l K_{h_1}(X_{1l}-x_1)V_{kl} \right)\nonumber \\
&=&n^{-4}\sum_{i\neq j}\sum_{k\neq l} \E \left( \epsilon_jK_{h_1}(X_{1j}-x_1)\E_j(V_{ij}\,\vert\, X_j, X_k, X_l)\epsilon_l K_{h_1}(X_{1l}-x_1)V_{kl} \right)\nonumber \\
&\leq&C n^{-3}\sum_{j\neq \{k,l\}}\sum_{k\neq l} \E \left( \epsilon_jK_{h_1}(X_{1j}-x_1)\epsilon_l K_{h_1}(X_{1l}-x_1)\E_l(V_{kl}\,\vert\, X_j, X_l) \right)\nonumber \\
&\leq&C n^{-2}\sum_{j\neq l} \E \left( \epsilon_jK_{h_1}(X_{1j}-x_1) \epsilon_l K_{h_1}(X_{1l}-x_1) \right)\nonumber \\
&\leq&C n^{-1}\sum_{k=1}^{n-1} \Cov \left( \epsilon_1K_{h_1}(X_{11}-x_1), \epsilon_{k+1} K_{h_1}(X_{1k+1}-x_1) \right)\nonumber \\
&=&o(n^{-1}h_1^{-p})\nonumber 
\end{eqnarray}
In the last Equation, we used Equations \eqref{eq:trick1} - \eqref{eq:trick2}. 
Next, for three different indices $k\in \{i,j\}$, $l\not\in \{i,j\}$
\begin{eqnarray}
&&  n^{-4} \sum_{i\neq j} \sum_{l\neq \{i,j\}} \E \left(\epsilon_jK_{h_1}(X_{1j}-x_1)V_{ij}  \epsilon_lK_{h_1}(X_{1l}-x_1)V_{kl} \right) \nonumber\\
&\leq&  C n^{-4} h_2^{-q}  \sum_{i\neq j} \sum_{l\neq \{i,j\}} \E \left( \epsilon_{j}  K_{h_1}(X_{1j}-x_1) \E_j (V_{ij}\,\vert\, X_j, X_l) \epsilon_{l} K_{h_1}(X_{1l}-x_1) \right) \nonumber\\
&\leq&  C n^{-3} h_2^{-q}  \sum_{j\neq l} \Cov \left(\epsilon_{j} K_{h_1}(X_{1j}-x_1) , \epsilon_{l} K_{h_1}(X_{1l}-x_1) \right) \nonumber\\
&\leq&C n^{-2}h_2^{-q}\sum_{k=1}^{n-1} \Cov \left( \epsilon_1K_{h_1}(X_{11}-x_1), \epsilon_{k+1} K_{h_1}(X_{1k+1}-x_1) \right)\nonumber \\
&=&  o(n^{-2}h_1^{-p}h_2^{-q})=  o(n^{-1}) \nonumber
\end{eqnarray}
Also, for $k\not\in \{i,j\}$, $l=i$
\begin{eqnarray}
&&  n^{-4} \sum_{i\neq j} \sum_{k\neq \{i,j\}} \E \left(\epsilon_jK_{h_1}(X_{1j}-x_1)V_{ij}  \epsilon_i K_{h_1}(X_{1i}-x_1)V_{ki}  \right) \nonumber\\
&\leq&  C n^{-4}  h_2^{-q}\sum_{i\neq j} \sum_{k\neq \{i,j\}} \E \left(\epsilon_jK_{h_1}(X_{1j}-x_1)  \epsilon_i K_{h_1}(X_{1i}-x_1) \E_i (V_{ki}\,\vert\,X_i,X_j)  \right) \nonumber\\
&\leq& C n^{-3}  h_2^{-q}\sum_{i\neq j} \Cov \left(  \epsilon_{j} K_{h_1}(X_{1j}-x_1),  \epsilon_{i}  K_{h_1}(X_{1i}-x_1)\right) \nonumber\\
&\leq&C n^{-2}h_2^{-q}\sum_{k=1}^{n-1} \Cov \left( \epsilon_1K_{h_1}(X_{11}-x_1), \epsilon_{k+1} K_{h_1}(X_{1k+1}-x_1) \right)\nonumber \\
&=&  o(n^{-2}h_1^{-p}h_2^{-q})=  o(n^{-1}) \nonumber
\end{eqnarray}
And for $k\not\in \{i,j\}$, $l=j$
\begin{eqnarray}
&&  n^{-4} \sum_{i\neq j} \sum_{k\neq \{i,j\}} \E \left(\epsilon^2_jK^2_{h_1}(X_{1j}-x_1)V_{ij} V_{kj} \right) \nonumber\\
&=&  n^{-4} \sum_{i\neq j} \sum_{k\neq \{i,j\}} \E \left(\epsilon^2_jK^2_{h_1}(X_{1j}-x_1)\E_j (V_{ij}\,\vert\,X_k, X_j) V_{kj} \right) \nonumber\\
&\leq&C  n^{-3} \sum_{j\neq k} \E \left(\epsilon^2_jK^2_{h_1}(X_{1j}-x_1) \E_j (V_{kj}\,\vert\,X_j) \right) \nonumber\\
&\leq&  Cn^{-2}  \sum_{j} \E \left(\epsilon^2_{j} K_{h_1}^2(X_{1j}-x_1) \right) \nonumber\\
&=& C n^{-1} \E \left( \epsilon^2_{1} K_{h_1}^2(X_{11}-x_1) \right) \nonumber\\
&=& O(n^{-1}h_1^{-p})\nonumber
\end{eqnarray}
Here, we used Equation \eqref{eq:kernelvar}. Furthermore, for two different indices $k=i$, $l=j$
\begin{eqnarray}
&&  n^{-4} \sum_{i\neq j}  \E \left(\epsilon^2_jK^2_{h_1}(X_{1j}-x_1)V^2_{ij} \right) \nonumber\\
&=&  n^{-4} \sum_{i\neq j}  \E \left(\epsilon^2_jK^2_{h_1}(X_{1j}-x_1)\E_j(V^2_{ij}\,\vert\,X_j) \right) \nonumber\\
&\leq&  C n^{-3}h_2^{-q}   \sum_{j} \E \left( \epsilon^2_{j} K_{h_1}^2(X_{1j}-x_1) \right) \nonumber\\
&=&  C n^{-2}h_2^{-q}\E \left( \epsilon^2_{1} K_{h_1}^2(X_{11}-x_1) \right) \nonumber\\
&=& O(n^{-2}h_1^{-p}h_2^{-q})=o(n^{-1})\nonumber
\end{eqnarray}
Again using Equation \eqref{eq:kernelvar}. Lastly, $k=j$, $l=i$
\begin{eqnarray} 
&&  n^{-4} \sum_{i\neq j}  \E \left(\epsilon_{j} K_{h_1}(X_{1j}-x_1) V_{ij} \epsilon_{i} K_{h_1}(X_{1i}-x_1) V_{ji}  \right) \nonumber\\
&\leq& C n^{-4}  h_2^{-2q} \sum_{i\neq j}\E \left(\epsilon_{j} K_{h_1}(X_{1j}-x_1) \epsilon_{i} K_{h_1}(X_{1i}-x_1) \right) \nonumber\\
&=&  C n^{-4}  h_2^{-2q} \sum_{i\neq j}\Cov \left( \epsilon_{j} K_{h_1}(X_{1j}-x_1), \epsilon_{i} K_{h_1}(X_{1i}-x_1) \right) \nonumber\\
&\leq&C n^{-3}h_2^{-2q}\sum_{k=1}^{n-1} \Cov \left( \epsilon_1K_{h_1}(X_{11}-x_1), \epsilon_{k+1} K_{h_1}(X_{1k+1}-x_1) \right)\nonumber \\
&=& o(n^{-3}h_1^{-p}h_2^{-2q})=o(n^{-1})\nonumber
\end{eqnarray}
Combining above rates yields $\E T_{n,3}^2=O(n^{-1}h_1^{-p})$.
Even though the previous calculations were limited to $T_{n,3}$,  $\E T_{n,4}^2=O(n^{-1}h_1^{-p})$ follows in exactly the same manner by replacing $V_{ij}$ with $L_{h_2}(X_{2j}-X_{2i})$ and using  $\E_j(L_{h_2}(X_{2j}-X_{2i})\,\vert\, X_j, X_k, X_l)=O(1)$,  $\E_j(L^2_{h_2}(X_{2j}-X_{2i})\,\vert\, X_j)=O(h_2^{-q})$ almost surely (see Lemma~\ref{lemma1}). Together with Remark 3 \cite{fan1998}, this concludes the proof of Theorem~\ref{thm1}.
\end{proof}

\subsection{Lemmata}\label{app:moments}

\begin{lemma}\label{lemma1}
	$\E_j(V_{ij}\,\vert\, X_j, X_k, X_l)=O(1)$, $\E_j(V^2_{ij}\,\vert\, X_j)=O(h_2^{-q})$, $\E_j \left(L_{h_2}(X_{2j}-X_{2i})\,\vert\,X_j, X_k, X_l\right)=O(1)$ and $\E_j \left(L^2_{h_2}(X_{2j}-X_{2i})\,\vert\, X_j\right)=O(h_2^{-q})$ almost surely.
\end{lemma}
\begin{proof}

\begin{eqnarray}
\E_j(V_{ij}\,\vert\, X_j, X_k, X_l)&=& \int \Gamma(X_{2i})L_{h_2}(X_{2j}-X_{2i})p(X_{2i}\,\vert\, X_j, X_k, X_l)dX_{2i} - p(X_{2j})\Gamma(X_{2j})\nonumber \\
&\leq&C \int \Gamma(X_{2j}+h_2u)L(u)du - p(X_{2j})\Gamma(X_{2j})\nonumber\\
&=& C\Gamma(X_{2j})- p(X_{2j})\Gamma(X_{2j})+ o(1)=O(1).\nonumber
\end{eqnarray}

\begin{eqnarray}
\E_j(V_{ij}^2\,\vert\,X_j)&=& \E_j \left(\Gamma^2(X_{2i}) L^2_{h_2}(X_{2j}-X_{2i})-2p(X_{2j})\Gamma(X_{2j})\Gamma(X_{2i})L_{h_2}(X_{2j}-X_{2i})\right. \nonumber\\
&+&\left. p^2(X_{2j})\Gamma^2(X_{2j})\,\vert\,X_j\right)\nonumber\\
&=& \E_j \left(\Gamma^2(X_{2i}) L^2_{h_2}(X_{2j}-X_{2i})\,\vert\,X_j\right)- 2p(X_{2j})\Gamma(X_{2j})\int\Gamma(X_{2i})L_{h_2}(X_{2j}-X_{2i})\nonumber\\
&&p(X_{2i}\,\vert\,X_j)dX_{2i}+p^2(X_{2j})\Gamma^2(X_{2j})\nonumber\\
&\leq& \E_j \left(\Gamma^2(X_{2i}) L^2_{h_2}(X_{2j}-X_{2i})\,\vert\,X_j\right)+C \vert p(X_{2j})\Gamma(X_{2j})\int\Gamma(X_{2j}+h_2u)L(u)du\vert\nonumber\\
&+&p^2(X_{2j})\Gamma^2(X_{2j})\nonumber\\
&=& \E_j \left(\Gamma^2(X_{2i}) L^2_{h_2}(X_{2j}-X_{2i})\,\vert\,X_j\right)+C p(X_{2j})\Gamma^2(X_{2j})+p^2(X_{2j})\Gamma^2(X_{2j})+o(1)\nonumber\\
&=& \E_j \left(\Gamma^2(X_{2i}) L^2_{h_2}(X_{2j}-X_{2i})\,\vert\,X_j\right)+O(1)\nonumber\\
&=& O(h_2^{-q}),\nonumber
\end{eqnarray}
where we used
\begin{eqnarray}
&&\E_j \left(\Gamma^2(X_{2i}) L^2_{h_2}(X_{2j}-X_{2i})\,\vert\,X_j\right)\nonumber \\
&=& \int \Gamma^2(X_{2i}) L^2_{h_2}(X_{2j}-X_{2i}) p(X_{2i}\,\vert\,X_j)dX_{2i} \nonumber \\
&\leq&C h_2^{-q} \int \Gamma^2(X_{2j}+h_2u) L^2(u)du \nonumber \\
&=&Ch_2^{-q} \Gamma^2(X_{2j}) \int L^2(u)du + o(h_2^{-q})=O(h_2^{-q}). \nonumber 
\end{eqnarray}
Above calculations imply $\E_j \left(L_{h_2}(X_{2j}-X_{2i})\,\vert\,X_j, X_k, X_l\right)=O(1)$ and $\E_j \left(L^2_{h_2}(X_{2j}-X_{2i})\,\vert\,X_j\right)=O(h_2^{-q})$ almost surely.
\end{proof}

\begin{lemma}\label{lemma2}
$\E A_j(x^i)\hat{r}_{ij}=O(h_1^2)$.
\end{lemma}
\begin{proof}
\begin{eqnarray}
\E A_j(x^i)\hat{r}_{ij}&=& \E K_{h_1}(X_{1j}-x_1)L_{h_2}(X_{2j}-X_{2i})[m(X_j)-m(x^i)-f'_1(x_1)^T(X_{1j}-x_1)] \nonumber 
\end{eqnarray}
Note, that
\begin{eqnarray}
	&&\E_j( L_{h_2}(X_{2j}-X_{2i})m(x^i)\,\vert\,X_j) \nonumber\\
	&=&\int L_{h_2}(X_{2j}-X_{2i})m(x_1,X_{2i})p(X_{2i}\,\vert\,X_j) dX_{2i} \nonumber \\
	&=&\int L(u)m(x_1,X_{2j}+h_2u)p(X_{2j}+h_2u\,\vert\,X_j) du \nonumber \\
	&=&m(x_1,X_{2j}) p(X_{2j}\,\vert\,X_j) + O(h_2^d) \nonumber 
\end{eqnarray}
almost surely. Here, we used that $L(\cdot)$ is a order $d$-kernel by Assumption~1.7.
Similarly,
\begin{equation}
	\E_j( L_{h_2}(X_{2j}-X_{2i})\,\vert\,X_j) 
	= p(X_{2j}\,\vert\,X_j) + O(h_2^d) \nonumber 
\end{equation}
almost surely. By Assumption~1.7, $K(\cdot)$ is a symmetric kernel. Hence,
\begin{equation}
	\E( K_{h_1}(X_{1j}-x_{1})) 
	= p(x_{1}) + O(h_1^2). \nonumber 
\end{equation}
Then,
\begin{eqnarray}
&&\E K_{h_1}(X_{1j}-x_1)L_{h_2}(X_{2j}-X_{2i})[m(X_j)-m(x^i)] \nonumber \\
&=&  \E K_{h_1}(X_{1j}-x_1)p(X_{2j}\,\vert\,X_{1j},X_{2j})[m(X_j)-m(x^j)]+O(h_2^d)\nonumber\\
&=&  \int \int  K_{h_1}(X_{1j}-x_1)[m(X_j)-m(x^j)] p(X_{2j}\,\vert\,X_{1j},X_{2j})p(X_{1j},X_{2j})dX_{1j}dX_{2j}+O(h_2^d)\nonumber\\
&=&  \int \int  K(u)[m(x_1+h_1u,X_{2j})-m(x_1,X_{2j})]du \,p(X_{2j}) dX_{2j}+O(h_2^d)\nonumber \\
&=&  \int [ m(x_1,X_{2j})-m(x_1,X_{2j})]p(X_{2j}) dX_{2j} +O(h_1^2) +O(h_2^d)\nonumber\\
&=&O(h_1^2). \nonumber
\end{eqnarray}
In the last line, we used the bandwidth condition $h_2^d=o(h_1^2)$. We strengthen Assumption~A.4 \cite{fan1998} such that  $f_1(\cdot)$ has bounded first and second derivative. Hence,
\begin{eqnarray}
&&\E K_{h_1}(X_{1j}-x_1)L_{h_2}(X_{2j}-X_{2i})f'_1(x_1)^T(X_{1j}-x_1)\nonumber \\
&\leq &C \E K_{h_1}(X_{1j}-x_1)(X_{1j}-x_1)\nonumber \\
&= &C \int h_1 uK(u)p(x_1+h_1u)du\nonumber \\
&= &C h_1p(x_1)\int uK(u)du+O(h_1^2)=O(h_1^2).\nonumber 
\end{eqnarray}

\end{proof}

\section{Additional simulations}\label{app:simulations}

\subsection{The choice of the bandwidth $h$}
\label{app:choiceh}
For $\ARCH$- and $\GARCH$-models, the MSE behaves in a different manner than for the models in Section~\ref{section:choiceh}. As explained in Section~\ref{section:choiceh}, the causal effect is identical to zero in $\ARCH$-models. $\GARCH$-models can be rewritten as $\ARCH(\infty)$. Therefore, the true Markovian order $p_0$ is infinity and the causal effect equals zero as well. In both cases, it is beneficial to choose a large bandwidth even if no boosting iterations are performed. The reason seems to be that estimating the zero causal effect function can be well done when choosing a large bandwidth: the estimator then approximates the mean of the underlying time series which is equal to zero as well. Boosting does not improve the initial kernel estimate. For large bandwidths, however, the difference between the estimate with and without boosting is negligible as shown in Figure~\ref{fig:bandwidth3}.

\begin{figure}[!htb]
       \centering
          \begin{subfigure}[b]{0.7\textwidth}
        \includegraphics[width=\textwidth]{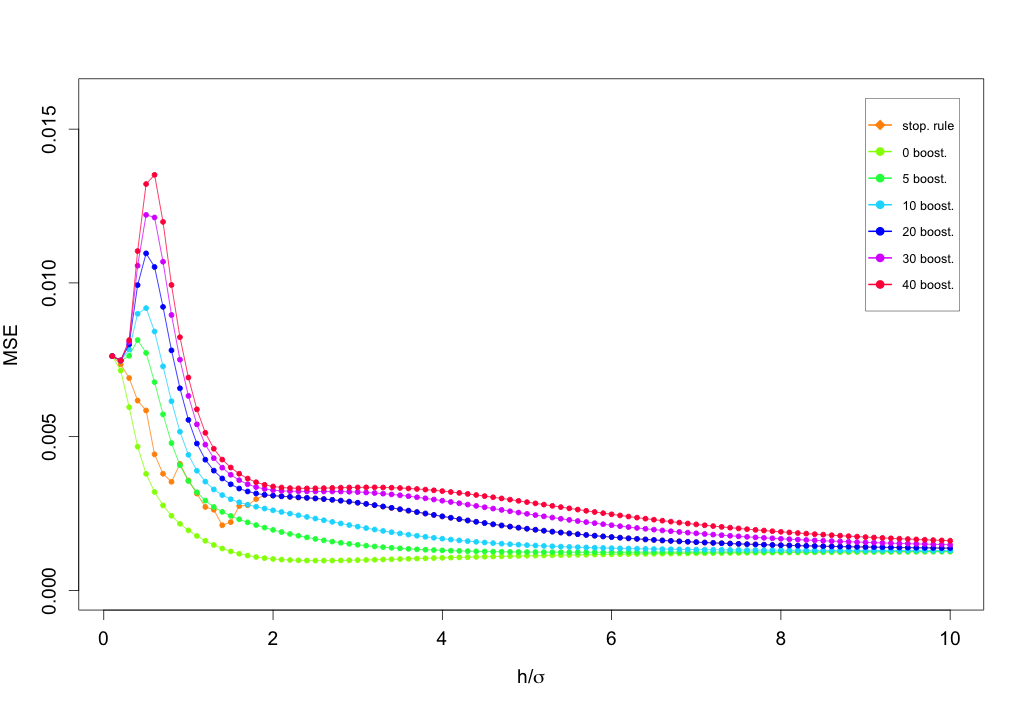}
        \caption{Model 3}
        \label{fig:model3}
    \end{subfigure}
    
    \begin{subfigure}[b]{0.7\textwidth}
        \includegraphics[width=\textwidth]{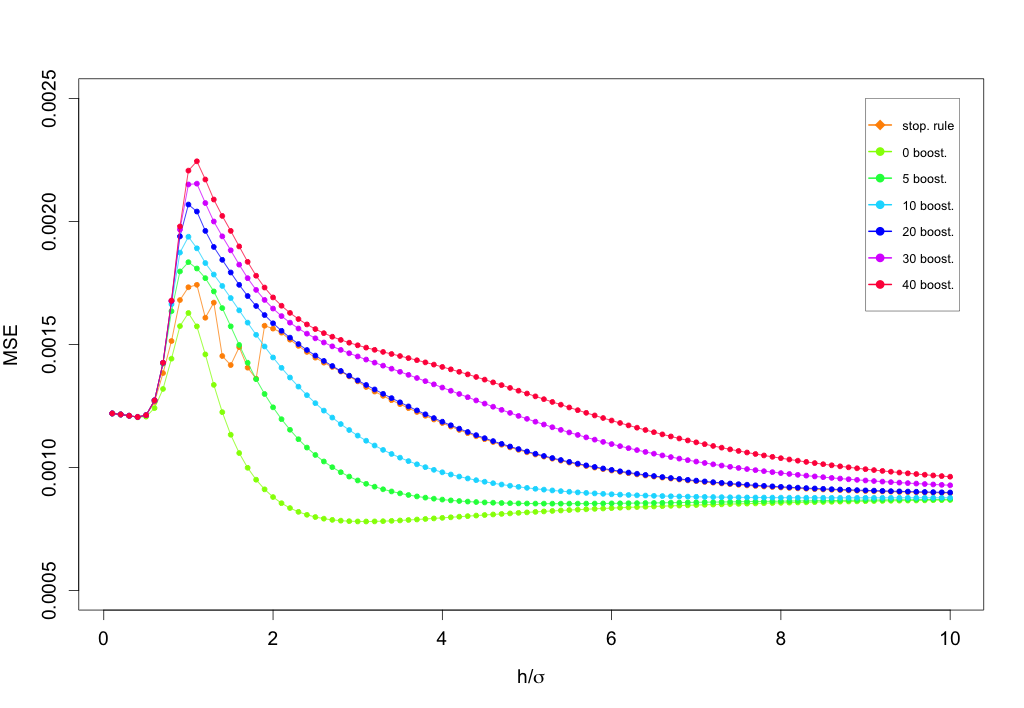}
        \caption{Model 4}
        \label{fig:model4}
    \end{subfigure}

    \caption{
      Dependence on bandwidth $h$. MSE values, for models 3 and 4, for $h$
   between $0.1\hat{\sigma}$ 
   and $10\hat{\sigma}$ ($x$-axis with scaled $h/\hat{\sigma}$) for \mbox{MINT-T},
   without boosting, with fixed number of 
   boosting iterations, and with stopping rules. The time lag $p$ is set to
      4 on model 3 and 10 on model 4.}\label{fig:bandwidth3} 
\end{figure}

\clearpage

\bibliography{mybib}
\bibliographystyle{plainnat}

\end{document}